\documentclass{article}
\usepackage{graphicx} 
\usepackage{amsmath}
\usepackage{amsthm}
\usepackage{amsfonts}
\usepackage{amssymb}
\usepackage{natbib}
\usepackage{xcolor}
\usepackage{hyperref}
\usepackage{mathtools}
\usepackage{algpseudocode}
\usepackage{algorithm}

\usepackage{geometry}
\usepackage{array}
\usepackage{longtable}

\newcommand{\bb}[1]{\boldsymbol{#1}}

\newcommand{\Ines}[1]{\textcolor{black}{#1}}
\newcommand{\Pat}[1]{\textcolor{black}{#1}}

\newtheorem{proposition}{Proposition}
\newcommand{\nicefrac}[2]{\frac{#1}{#2}}

\author{
  Puchhammer, Patricia\\
  \texttt{patricia.puchhammer@tuwien.ac.at}
  \and
  Wilms, Ines\\
  \texttt{i.wilms@maastrichtuniversity.nl}
  \and
  Filzmoser, Peter\\
  \texttt{peter.filzmoser@tuwien.ac.at}
}

\title{Sparse outlier-robust PCA for multi-source data}

\begin{document}
\maketitle
\begin{abstract}

Sparse and outlier-robust Principal Component Analysis (PCA) has been a very active field of research recently. Yet, most existing methods apply PCA to a single dataset whereas multi-source data--i.e.\ multiple related datasets requiring joint analysis--arise across many scientific areas. We introduce a novel PCA methodology that simultaneously (i) selects important features, (ii) allows for the detection of global sparse patterns across multiple data sources as well as local source-specific patterns, and (iii) is resistant to outliers. To this end, we develop a regularization problem with a penalty that accommodates global-local structured sparsity patterns, and where the ssMRCD estimator is used as plug-in to permit joint outlier-robust analysis across multiple data sources. We provide an efficient implementation of our proposal via the Alternating Direction Method of Multiplier and illustrate its practical advantages in simulation and in applications.
\end{abstract}

\section{Introduction} 


Principal component analysis (PCA) is undoubtedly one of the most important unsupervised statistical methods available. 
The basic idea is to project the observations in a given dataset onto a new vector space with orthonormal basis where each basis vector is a linear combination of the original variables constructed to capture the highest variability for the first basis vector, the second highest variability for the second basis vector and so on. The new variables are called \textit{Principal Components} (PC), the coordinates of the PCs in the original variable space are called \textit{loadings} and the coordinates of the observations with respect to the PCs are called \textit{scores}. Often, only the first few PCs that catch a majority of the variance and thus of the available information are analyzed. As such, PCA finds widespread application across numerous areas, such as dimensionality reduction, visualization, clustering, feature engineering and many more.

For standard PCA we get loadings that are often a combination of all variables involved. Especially nowadays with datasets consisting of many variables, sensible, efficient and correct interpretation of scores and loadings can get difficult. Moreover, by implicitly (or also explicitly) focusing the interpretation on large (absolute) loading entries and ignoring small ones, misleading interpretation results can be produced as discussed in \citet{Cadima1995}. Therefore, induced sparsity in the loading entries is necessary to ensure correct interpretation of PCA results. The literature on sparse PCA is rich. Starting with the work of \citet{Jolliffe2003}, who introduced the LASSO into PCA with the algorithm known as SCoTLASS, \citet{Zou2006} included an elastic-net penalty to PCA reformulated as regression problem. Further developments include the work of \citet{Shen2008} approaching the problem from a regularized singular value decomposition, \citet{Ma2013} focusing on a thresholding approach for high-dimensional data, \citet{Aspremont2008} deriving a greedy algorithm based on a semi-definite relaxation variation, and \citet{Journee2010} with a convex reformulation of sparse PCA.


While PCA or sparse PCA is often applied to datasets without further adjustments, many modern applications entail multiple related datasets \Pat{from different sources} for which PCA needs to be performed jointly. Classical examples where joint PCA is relevant are time series data that can be grouped based on time increments like days, months or years, spatial data with groups based on spatial proximity or nationality, or more general subgroups based on e.g.~demographics, socioeconomic status or other external variables. Even though PCA might still be applied globally on the whole data and structural changes might still be identified in the scores, the question of which variables drive the variance in different groups or datasets remains unanswered. However, fully localizing PCA by applying it on each subgroup individually ignores the overall inert link between the subgroups, rendering the individualistic approach inappropriate. Thus, the local and global aspect of the data needs to be addressed simultaneously. Moreover, in the multi-source PCA setting, sparsity and especially structured sparsity patterns are well suited. By analyzing multiple related datasets, we end up with $N$-times more loading entries than for a global PCA approach. Thus, sparsity in each entry is important. However, due to the interconnection of the datasets and followingly of the loadings, \Pat{structured sparsity, here meaning sparsity in entries of the same variable for all sources simultaneously, can be present in the datasets as well. Including a structured sparsity inducing combination of groupwise and elementwise penalties} can increase accuracy in PCA or also regression results as demonstrated by \citet{Jenatton2010} and \citet{Simon2013}, \Pat{respectively, for groupings of variables in a global context.} Although other disciplines have already explored multi-group aspects successfully, \Pat{for example \citet{Puchhammer2023} for robust and \citet{Danaher2014} and \citet{Price2021} for inverse sparse covariance estimation, \citet{Barbaglia2016} and \citet{Wilms2018} for time series data or \citet{Price2018} and \citet{Wang2013} in a multivariate regression context, it is still underexplored in the context of PCA.}


Especially when it comes to variability, outliers must be taken care of reliably. By definition, outliers do not behave like the majority of the data and lie outside of the multivariate point cloud of regular observations. Thus, outliers inherently increase variability measured by classical estimators and distort the direction of high variability towards them. Since we are interested in the direction of variability of the data majority, robustness in estimators for variability, i.e. covariance matrices, must be used. Well-known robust covariance estimators are for example the \textit{Minimum Covariance Determinant estimator} (MCD, \citet{Rousseeuw1985, Rousseeuw1999}) or its regularized variant, the \textit{Minimum Regularized Covariance Determinant estimator} (MRCD, \citet{Boudt2020}) that can be used to robustify PCA by a so-called \textit{plug-in} approach \citep{Croux2000}. Other approaches are based on projection pursuit \citep[e.g.][]{Li1985,  Hubert2002, Croux2005} or a combination of both called ROBPCA \citep{Hubert2005}. \citet{Hubert2016} further extend ROBPCA to sparse PCA (ROSPCA) while \citet{Croux2013} develop a robust PCA method with standard sparsity based on projection pursuit. However, no research has been done on robust PCA with structured sparsity patterns necessary in a setting with multiple groups
of observations.


We fill the gap in this paper with an intuitive extension of robust PCA to jointly analyze related datasets. Based on a robust covariance estimator that jointly calculates covariance matrices for multiple groups of data by induced smoothing,  called \textit{spatially smoothed MRCD} (ssMRCD) estimator \citep{Puchhammer2023, ssMRCD_Cran}, we apply the plug-in approach for robust PCA. Standard sparsity as well as structured sparsity penalties are included to mirror the relations between the \Pat{multiple sources (also called \textit{neighborhoods} in the ssMRCD context)}. By jointly analyzing the covariance matrix and sparsity in the loadings we can better differentiate between global structures indicated by similarities between sources, and local structures indicated by differences. 

In the remainder of the paper we will derive the objective function for sparse PCA in Section~\ref{sec:problem} and derive a numerical solution based on the Alternating Direction Method of Multipliers in Section~\ref{sec:algorithm}. In Section~\ref{sec:simulations} the algorithm is tested and the method is compared to ROSPCA and classical PCA. Two real data examples are analyzed in Section~\ref{sec:realdata} and finally, conclusions are given in Section~\ref{sec:conclusion}.




\section{Multi-source PCA based on the ssMRCD}
\label{sec:problem}

\Pat{In Section~\ref{subsec:firstPC} we introduce the minimization problem for the first PC. We expand the problem to further components in Section~\ref{subsec:multiplePC} and explain the plug-in approach with the ssMRCD estimator in Section~\ref{subsec:ssMRCD} in more detail}.

\subsection{First Principal Component}
\label{subsec:firstPC}

Let $\bb{X} = \left( \bb{x}_1', \ldots, \bb{x}_n' \right)' \in \mathbb{R}^{n \times p}$ be a data matrix of $n$ observations $\bb{x}_1, \ldots, \bb{x}_n$ with $p$ variables. \Pat{The observations are partitioned into~$N$ sources~$a_i$, $i = 1, \ldots, N$ with corresponding locally estimated covariance matrices denoted by~$\bb{\hat{\Sigma}}_i \in \mathbb{R}^{p \times p}$ and estimated means~$\bb{\hat{\mu}}_i \in \mathbb{R}^{p}$.}

When considering the loadings in a setting with multiple covariance matrices, each loading can be written as a matrix, where each column represents a source, and each row one variable. The loadings matrix of the $k$-th principal component is denoted as
\begin{align}
\label{eq:notation_matrix}
     \bb{V}^k = \left( 
    \begin{array}{ccc}
         v_{11}^k&  \ldots & v_{1N}^k\\
         \vdots &  & \vdots \\
         v_{p1}^k & \ldots & v_{pN}^k
    \end{array}
    \right) = (\bb{v}_{\cdot 1}^k, \ldots, \bb{v}_{\cdot N}^k) = (\bb{v}_{1\cdot}^{k'}, \ldots, \bb{v}_{p\cdot}^{k'})' \ .
\end{align}
The loadings matrix of the first PC is obtained by solving the following optimization problem
\begin{align}
\label{eq:objective_matrix}
   \bb{V}^1 = \underset{\substack{\bb{V} \in \mathbb{R}^{p \times N} \\ ||\bb{v}_{\cdot i}||_2 = 1, i = 1, ..., N}}{\text{argmin}} -\sum_{i=1}^N \bb{v}_{\cdot i}' \bb{\hat{\Sigma}}_i \bb{v}_{\cdot i} + \eta \gamma \sum_{j = 1}^{p} \underbrace{\sum_{i = 1}^{N} |v_{ji}|}_{ = ||\bb{v}_{j\cdot}||_1} + \eta (1-\gamma) \sqrt{N}  \sum_{j= 1}^p ||\bb{v}_{j \cdot} ||_2,
\end{align}
where $\eta \ge 0$ regularizes the degree of sparsity, and $\gamma \in [0,1]$ distributes the \Pat{sparsity between local ($\gamma = 1$) and global ($\gamma = 0$) sparsity patterns. The groupwise penalty induces global sparsity structures and is equivalent to the groupwise penalty used in \citet{Simon2013}}. As in \citet{Simon2013}, the term ~$\sqrt{N}$ balances the size of the two penalties since the minimal penalty for the L$_1$-norm under the given constraints is~$N$, whereas the minimal groupwise penalty is $\sqrt{N}$. Thus, we can compare the effect of increasing~$\eta$ among different levels of~$\gamma$ more easily. 

\subsection{Multiple Principal Components}
\label{subsec:multiplePC}

The loadings for the $k$-th principal component $\bb{V}^k$ are the solutions to the optimization problem Equation~\eqref{eq:objective_matrix} with an additional constraint to account for orthogonality per source,
\begin{align}
    \bb{v}_{\cdot i}^k \bot \bb{v}_{\cdot i}^l \hspace{0.5cm} l = 1, \ldots, k-1,  i = 1, \ldots, N.
\end{align}
The orthogonality constraints between the loadings per source constitute non-standard optimization constraints, especially in the context of (sparse) PCA. Since the groupwise sparsity induces a non-separable objective function and existing solutions rely on standard orthogonality constraints, they cannot be applied, and new solutions are needed. 

To facilitate notation and optimization, we can rewrite the problem into a stacked-column notation. The matrix~$\bb{V}^k$ can be stacked into one collective vector~$\bb{v}^k$ and the covariance matrices into a block-diagonal (positive semi-definite) matrix,
\begin{align*}
    \bb{v}^k= \left( \begin{array}{c}
         \bb{v}_{\cdot 1}^k \\
          \vdots  \\
         \bb{v}_{\cdot N}^k 
    \end{array}
    \right), \hspace{0.5cm}
    \bb{\hat{\Sigma}} = \left( 
    \begin{array}{ccc}
         \bb{\hat{\Sigma}}_1 & &   \\
         & \ddots  &\\ 
          & & \bb{\hat{\Sigma}}_N
    \end{array}
    \right).
\end{align*}
Then, the objective function for the loadings of the~$k$-th set of principal components can be rewritten in a form known from standard PCA with adapted penalty terms and linear and quadratic equality and inequality constraints,
\begin{align}
\label{eq:objective}
     \bb{v}^k = \underset{\substack{\bb{v} \in \mathbb{R}^{pN} }}{\text{argmin}} \hspace{0.5cm} &  -\bb{v}'\bb{\hat{\Sigma}} \bb{v} + \eta \gamma  ||\bb{v}||_1 + \eta (1-\gamma) \sqrt{N} \sum_{j= 1}^p \sqrt{\bb{v}' \bb{C}_j \bb{v}}   \notag \\
    \text{s.t.} \hspace{0.5cm} &  \bb{v}'\bb{B}_i\bb{v} = 1  \hspace{0.5cm} \forall i = 1, \ldots, N   \\
     & \bb{v}'\bb{B}_i\bb{v}^l = 0 \hspace{0.5cm} \forall l = 1, \ldots, k-1, \ i = 1, \ldots, N. \notag
\end{align}
The $pN \times pN$ matrices $\bb{C}_j$ and $\bb{B}_i$,  extract the $j$-th row (variable) and $i$-th column (source) of~$\bb{V}^k$ from the stacked column vector~$\bb{v}^k$, respectively, and are defined as 
\begin{align*}
    (\bb{C}_j)_{ik} &= \left\{
    \begin{array}{cl}
         1, & \text{if }  i = k = pl + j, \text{ where } l=0,\ldots ,N-1,  \\
         0, & \text{otherwise.}
    \end{array}
    \right. \\ 
    (\bb{B}_i)_{jk} &= \left\{
    \begin{array}{cl}
         1, & \text{if }  j = k = p(i-1)+l, \text{ where } l=1,\ldots ,p,\\
         0, & \text{otherwise.}
    \end{array}
    \right.
\end{align*} 
Thus, $\bb{C}_j  = \bb{C}_j' = \bb{C}_j' \bb{C}_j$ and $\bb{C}_1 + \ldots + \bb{C}_p = \bb{I}_{pN}$ for $j = 1, \ldots, p$ and $\bb{B}_i  = \bb{B}_i' = \bb{B}_i' \bb{B}_i$ and $\bb{B}_1 + \ldots + \bb{B}_p = \bb{I}_{pN}$ for $i = 1, \ldots, N$.

Once the loadings $\bb{v}^1, \ldots, \bb{v}^k$ are obtained from the data, the scores of each locally centered observation $\bb{x}_\iota- \bb{\hat{\mu}}_{a(\iota)}$ for $\iota = 1, \ldots, n$ in source $a(\iota)$ are calculated by 
\begin{align}
    \label{eq:scores}
    \bb{t}_\iota = (\bb{x}_\iota-\bb{ \hat{\mu}}_{a(\iota)}) \left(\bb{v}^1_{\cdot a(\iota)}, \ldots, \bb{v}^k_{\cdot a(\iota)} \right)
\end{align}
and collected in $\bb{T} = \left( \bb{t}_1', \ldots, \bb{t}_n' \right)' \in \mathbb{R}^{n \times k}$.

\subsection{Outlier-robustness via ssMRCD Plug-in}
\label{subsec:ssMRCD}

The optimization of problem Equation~\eqref{eq:objective} for the loadings requires plug-in estimators for the covariance matrix of each source, whereas computation of the scores in Equation~\eqref{eq:scores} additionally requires mean estimators. Standard choices to this end would be the sample covariance matrices and sample means computed for each source separately. Such estimators face, however, two problems. First, they are not robust to outliers. One may resort to traditional robust estimators such as the MCD or median for each source separately to circumvent this problem. Since only robustly estimated covariances and means are used further, no additional robustification steps are necessary. Second, these classical (non-robust or robust) estimators still treat each source in isolation thereby ignoring potential connections and interactions between them. Therefore, it is crucial to incorporate both local and global information to leverage available information across multiple sources more extensively, enhancing the accuracy and reliability of the resulting covariance and mean estimators.

\Pat{An outlier-robust covariance and mean estimator tailored towards this global-local scenario is the ssMRCD estimator \citep{Puchhammer2023} that is implemented in the R-package \texttt{ssMRCD} \citep{ssMRCD_Cran}. The ssMRCD estimator has originally been developed for local outlier detection in spatial data. However, due to the quite general requirements it can easily be applied to more general datasets within a multi-source setting.}

\Pat{Starting with a partition of the data into multiple sources, the ssMRCD estimator selects a subset $H_i$ consisting of an $\alpha \in [0.5,1]$ percentage of observations of each source $i$ by minimizing the objective function over all H-subset combinations $\mathcal{H} = (H_i)_{i=1,\ldots ,N}$ 
\begin{equation*}
	\mathcal{H}^* = \underset{\substack{\mathcal{H} = (H_i)_{i=1,\ldots ,N}}}{\text{argmin}} \sum_{i = 1}^{N}  \det\left( (1-\lambda) \boldsymbol{K}_i(\mathcal{H}) + \lambda \sum_{j = 1 , j \neq i}^{N} \omega_{ij} \boldsymbol{K}_j(\mathcal{H}) \right),
\end{equation*}
similar to the MCD \citep{Rousseeuw1985, Rousseeuw1999}, or the MRCD \citep{Boudt2020} estimator, thus choosing subsets with least-outlying observations. The weight matrix~$\bb{W}$ should provide a measure of similarity between data sources which is used to leverage global information more targeted. For example, for spatial or time series data, the weights could be based on inverse distances, or for groupings based on known properties, the similarity between these properties' levels might be an appropriate choice for $\bb{W}$ (see also Section~\ref{sec:realdata}). The matrices $\boldsymbol{K}_i(\mathcal{H})$ are constructed in an MRCD manner \citep{Boudt2020}, regularizing the sample covariance $Cov(\boldsymbol{X}_{H_i})$ matrix of the H-subsets of source $i$ with a target matrix $T$ and a factor $\rho_i$
\begin{align*}
    \boldsymbol{K}_i(\mathcal{H})= \rho_i \boldsymbol{T} + (1-\rho_i) c_{\alpha} Cov(\boldsymbol{X}_{H_i}),
\end{align*} 
making the estimator useful also in for high-dimensional data. The consistency factor $c_{\alpha}$ is described in \citet{Croux1999}. Finally, the ssMRCD covariance estimators are then defined as
\begin{align*}
    \hat{\boldsymbol{\Sigma}}_{i} = (1-\lambda) \boldsymbol{K}_i(\mathcal{H}^*) + \lambda \sum_{j = 1 , j \neq i}^{N} \omega_{ij} \boldsymbol{K}_j(\mathcal{H}^*),
\end{align*}
and the mean estimators $\hat{\boldsymbol{\mu}}_{i}$ as the sample mean of the selected observations $\boldsymbol{X}_{H^*_i}$ per source.}

The most prominent parameter for the ssMRCD estimator $\lambda \in [0,1]$ defines the amount of smoothing between the covariances of sources weighted with $\bb{W}$. Thus, the amount of smoothing $\lambda$ describes how much of the global context is exploited compared to the local source specific data.

\section{Algorithm}
\label{sec:algorithm}

We present a computationally efficient algorithm tailored towards solving optimization problem~\eqref{eq:objective}. To this end, we use the \textit{Alternating Direction Method of Multipliers (ADMM)} \citep{Boyd2011}. The ADMM requires finding a reasonable decomposition into subproblems and linear constraints, which is equivalent to the given problem, and solving these decompositions iteratively until convergence. The approach uses key ideas of the \textit{Augmented Lagrangian} methods, that add a penalty term of deviations of the linear constraints into the objective function with a penalty weight $\rho > 0$. This should improve other ascent methods, e.g., dual ascent methods, and fewer assumptions for the original objective function are necessary~\citep[Section 2-3]{Boyd2011}.

We start by rewriting problem~\eqref{eq:objective} for the $k$-th PC as a consensus optimization problem \citep[Chapter 7.1]{Boyd2011} over four variables $\bb{v}_{(0)}, \bb{v}_{(1)}, \bb{v}_{(2)}$ and $\bb{v}_{(3)}$ as
\begin{align*}
    \min_{\bb{v}_{(1)}, \bb{v}_{(2)}, \bb{v}_{(3)}, \bb{v}_{(0)}} \hspace{0.5cm} & \underbrace{- \bb{v}_{(1)}'\bb{\hat{\Sigma}} \bb{v}_{(1)} + I_\infty\{ \bb{v}_{(1)}'\bb{B}_i\bb{v}_{(1)} = 1, \ \bb{v}_{(1)}'\bb{B}_i\bb{v}^l = 0  \ \forall 1 \le i \le N, 1 \leq l < k \}}_{f_1(\bb{v}_{(1)})} \\
    & + \underbrace{\eta \gamma  ||\bb{v}_{(2)}||_1}_{f_2(\bb{v}_{(2)})} + \underbrace{\eta (1-\gamma) \sqrt{N}\sum_{j= 1}^p \sqrt{\bb{v}_{(3)}' \bb{C}_j \bb{v}_{(3)}}}_{f_3(\bb{v}_{(3)})} \\
    \text{s.t.} \hspace{0.5cm} & \bb{v}_{(i)} - \bb{v}_{(0)} = 0, \hspace{0.5cm} i = 1, 2, 3,
\end{align*}
where $I_\infty$ denotes the indicator function with an infinite amount of weight if the condition inside the brackets is not fulfilled. 

The ADMM updates for the $m$-th iteration step are 
\begin{align}
\label{eq:ADMMsteps}
    \bb{v}_{(i)}^{m+1} = & \arg\min_{\bb{v}_{(i)}} (f_i(\bb{v}_{(i)}) + \frac{\rho}{2}|| \frac{1}{\rho} {\bb{u}_{(i)}^m} + \bb{v}_{(i)} - \bb{v}_{(0)}^m ||^2_2 \\
    \bb{v}^{m+1}_{(0)} = &\frac{1}{3}\sum_{i = 1}^{3} \left(\bb{v}_{(i)}^{m+1}  + \frac{1}{\rho}\bb{u}_{(i)}^m \right), \notag \\
    \bb{u}_{(i)}^{m+1} =&\bb{u}_{(i)}^{m} + \rho\left(\bb{v}_{(i)}^{m+1} - \bb{v}_{(0)}^{m+1} \right), \notag
\end{align}
with $\rho > 0$. For $\gamma = 1$, minor adaptations (i.e., removing the minimization problem for $\bb{v}_{(3)}$ and adapting appropriately) are made to facilitate achieving full sparsity for $\eta \rightarrow\infty$. The solutions of the three new optimization problems \eqref{eq:ADMMsteps} are discussed in Appendix~\ref{sec:app_subproblems}.

\Pat{In the following we describe topics of convergence in  Section~\ref{subsec:convergence} such as convergence criteria, starting values, the choice of the parameter $\rho$ and further enhancements. We also provide guidance to select the multiple hyperparameters in Section~\ref{subsec:hyperparameters}.}

\subsection{Convergence}
\label{subsec:convergence}
A globally optimal solution to problem~\eqref{eq:objective} exists since we have a compact variable space and a continuous objective function. Convergence of the iterative ADMM is based on the residual approach described in \citet[Chapter 7]{Boyd2011}. The primal and dual residuals are calculated as 
\begin{align*}
    r^m &= ||\bb{v}_{(1)}^m - \bb{v}_{(0)}^m||_2^2 + ||\bb{v}_{(2)}^m - \bb{v}_{(0)}^m||_2^2 + ||\bb{v}_{(3)}^m - \bb{v}_{(0)}^m||_2^2 \\
    s^m &= 3\rho^2 ||\bb{v}_{(0)}^m - \bb{v}_{(0)}^{m-1}||_2^2
\end{align*}
and should converge to zero. For a given $\epsilon_{ADMM}$, set to $10^{-4}$ in all simulations and real data examples, the tolerances for the two convergence criteria are
\begin{align}
\label{eq:eps_primedual}
    \epsilon_{prime}^m &= \sqrt{Np} \ \epsilon_{ADMM} + \epsilon_{ADMM} \max\{||\bb{v}_{(1)}^m||_2, ||\bb{v}_{(2)}^m|_2, ||\bb{v}_{(3)}^m||_2, ||\bb{v}_{(0)}^m||_2 \} \\
    \epsilon_{dual}^m  &
    = \sqrt{Np} \ \epsilon_{ADMM} + \epsilon_{ADMM} \max \{ ||\bb{u}_{(1)}^m||_2, ||\bb{u}_{(2)}^m||_2, ||\bb{u}_{(3)}^m||_2  \},
\end{align}
implying the same absolute and relative tolerance \citep[see also][]{Boyd2011} and the algorithm stops when $r^m < \epsilon_{prime}^m$ and $s^m < \epsilon_{dual}^m$.

Since we are interested in sparse loadings that often cannot be exactly zero due to the iterative nature of the algorithm, in our calculations we will round loading entries whose absolute values are below a tolerance of $\epsilon_{thr} = 5 \cdot 10^{-3}$. While a fixed tolerance seems somehow arbitrary, it enables fair comparisons between different parameter settings and possibly differing algorithm accuracies. Note that due to rounding small values to zero the orthogonality constraints might be slightly violated.

Note that neither the objective function nor the feasible subspace defined by the constraints in Equation~\eqref{eq:objective} is convex. Hence, there is neither a guarantee of obtaining a globally optimal solution nor convergence at all. Moreover, for standard or for sparse PCA, two vectors that are directly opposed to each other lead to the same explained variance and the same sparsity penalty. 
The same applies to our objective function, but on a smaller source level. \Pat{This is often not worth to mention since it is not necessarily an algorithmic problem \citep[see for example the SPCA algorithm of][]{Zou2006}.} 
However, \Pat{during experiments of prior stages of the algorithm} the ADMM has not been observed to converge but to switch between vectors, most likely due to the change per iteration of the penalty terms coming from the Augmented Lagrangian. To solve the convergence problem, we impose an additional constraint to the optimality problem in Equations~\eqref{eq:objective} to impede switching. Taking a fixed vector~$\bb{z} \in \mathbb{R}^{pN}$, we enforce the solution for the $k$-th vectorized loadings matrix to have a non-negative scalar product with~$\bb{z}$ for each source, 
\begin{align}
\label{eq:z}
    \bb{z}'\bb{B}_i\bb{v}^k \ge 0, \hspace{0.5cm} \forall i = 1, \ldots, N,
\end{align}
and add the term~$I_\infty\{\bb{z}'\bb{B}_i\bb{v}_{(1)} \ge 0 , \ \forall 1 \le i \le N\}$ to the objective function~$f_1(\bb{v}_{(1)})$.
Unless~$\bb{z}$ is exactly orthogonal to the real solution in at least one source, we consistently choose the solution for~$\bb{v}_{(1)}$ in the same direction. The best choice for $\bb{z}$ would be the true solution. Since the true solution is unknown, the next best possibility is the well-chosen starting value that we derive in Section~\ref{subsec:starting-value}. Section~\ref{subsec:alg_rho} further zooms into the choice of the penalty parameter $\rho$ and Section~\ref{subsec:alg_enhancements} ends with further algorithmic enhancements.

\subsubsection{Starting value} 
\label{subsec:starting-value}
The starting value plays a crucial role in convergence, especially due to the non-convexity of problem \eqref{eq:objective}. The closer the starting value is to the real solution of our problem, the faster the convergence. Also, the projection approach with the vector $\bb{z}$ that is chosen to be the starting value is more stable for a good choice of the starting value since the orthogonality issues described above are less likely to occur. 

To find a good starting value, a compromise between the two extreme cases of sparsity is needed. Our sparse PCA problem with no (additional) sparsity, $\eta = 0$, becomes separable and is the standard PCA problem per covariance $\bb{\hat{\Sigma}}_i$. Thus, the solution for the $k$-th PC are the $k$-th eigenvectors $\bb{y}_k(\bb{\hat{\Sigma}}_i)$ of the covariances calculated per source $a_i$, 
\begin{align*}
    \bb{y}_k^0 = (\bb{y}_k(\bb{\hat{\Sigma}}_1)', \ldots, \bb{y}_k(\bb{\hat{\Sigma}}_N)')'.
\end{align*}

The extreme solution for~$\eta \rightarrow \infty$ depends on~$\gamma$, $\bb{y}_k^{\infty}(\gamma)$, and can be calculated \Pat{based on the results of Proposition~\ref{proposition:extremes}. }

\begin{proposition}
\label{proposition:extremes}
Using the notation of Equation~\eqref{eq:objective_matrix}, define $G_1(\bb{v}) = \sum_{j = 1}^p \sum_{i = 1}^N |v_{ji}| = \sum_{i = 1}^N ||\bb{v}_{\cdot i}||_1$ and $G_2(\bb{v}) = \sum_{j = 1}^p ||\bb{v}_{j \cdot}||_2$. 
\begin{itemize}
    \item[a.] For each source $i = 1, \ldots, N$ and given the normality constraint $||\bb{v}_{\cdot i}||_2 = 1$, the minimal value of $||\bb{v}_{\cdot i}||_1$ is attained, if there exists a variable $j'(i)$ such that $|v_{j'(i)i}| = 1$ and $v_{ji} = 0$ for all $j \neq j'(i)$. 
    \item[b.] Given the normality constraints $||\bb{v}_{\cdot i}||_2 = 1, i = 1, \ldots, N$, the minimal value of $G_2(\bb{v})$ is attained if there exists a variable~$j'$ for all sources $i = 1, \ldots, N$ such that $|v_{j'i}| = 1$ and $v_{ji} = 0$ for all~$j \neq j'$.
    \item[c.] The minimizers of $G_1(\bb{v})$ with the highest explained variance $\sum_{i = 1}^N \bb{v}_{\cdot i}' \bb{\hat{\Sigma}_i} \bb{v}_{\cdot i}$ have corresponding indices for non-zero entries $j'(i) \in \mathrm{arg} \max_{j = 1, \ldots, p}  \left(\bb{\hat{\Sigma}_i}\right)_{jj}$ for each source $i$. 
    The minimizers of $G_2(\bb{v})$ with highest explained variance have non-zero entries only for the variable indexed by~$j' \in \mathrm{arg} \max_{j = 1, \ldots, p} \sum_{i = 1}^N \left(\bb{\hat{\Sigma}_i}\right)_{jj}$.
\end{itemize}
\end{proposition}
\begin{proof}
    The proof is given in Appendix~\ref{sec:app_startingvalues}.
\end{proof}

\Pat{Proposition~\ref{proposition:extremes} forms the basis for constructing the extreme solution set. Starting with the extreme solution of the first component, $\bb{y}_1^{\infty}(\gamma)$, when $\gamma = 1$, only $G_1$ is included in the penalty term, with penalty weight $\eta$ increasing without bounds. Thus, we focus solely on the minimizers of $G_1$ with the highest variance as an extreme solution. For $\gamma \neq 1$, we also need to minimize $G_2$. However, since the minimizers of $G_2$ are also minimizers of $G_1$, we ultimately seek minimizers of $G_2$ that explain most of the variance. }

\Pat{Secondly, the set of extreme solutions for subsequent PCs, is constructed iteratively. For the second and following components we need to make sure that the orthogonality constraints are maintained. Under the assumption that the prior PCs are extreme solutions as well with the previously described structure, every minimizer of the penalty term either per source for penalty $G_1$ or overall for penalty $G_2$ unequal to the prior extreme solutions satisfies the orthogonality constraints. Thus, we can focus on minimizers with the second highest, third highest, and so on, explained variance without further adjustments to the orthogonality constraints.}

\Pat{Note, that the indices in Proposition~\ref{proposition:extremes} are not necessarily unique, such as for correlation matrices, which are addressed in~Appendix~\ref{sec:app_startingvalues}.}

For a given~$\eta$ and~$\gamma$, we average the two extreme solutions to obtain an appropriate starting value, implicitly assuming some continuity of the solution in~$\eta$. Finally, we project the starting value onto the feasible subspace defined by the optimization constraints using the projection defined in Equation~\eqref{eq:projection}. Thus, the starting value for the $k$-th PC is 
\begin{align*}
\bb{y}_k = P_{\left(\bb{v}^{1:(k-1)}\right)^\bot}\left(\frac{1}{2} (\bb{y}_k^0 + \bb{y}_k^{\infty}(\gamma))\right).
\end{align*}

Note that the extreme solutions are based on orthogonality constraints using prior components, that are also derived by the extreme solution approach. Thus, the starting value/the extreme solutions are not applicable for just any vectors in the orthogonality constraints, but should be considered as a fixed set. \Pat{The good performance of the starting value is investigated in an additional simulation study in Appendix~\ref{subsec:app_startingvalues_simul}. }

good performance of starting value as investigated in additional simulation study in Appendix B1 or so

\subsubsection{Choice of penalty parameter $\rho$}
\label{subsec:alg_rho}
An open question connected to convergence is how to select the value of the penalty parameter $\rho$ as it vastly influences convergence and also convergence speed. Large values for $\rho$ focus on primal feasibility, thus keeping the primal residual small \citep{Boyd2011}. However, small primal residuals imply also smaller changes in the updates and possibly more iterations. The often-used approach to dynamically adapt $\rho$ based on the size of the residuals described in \citet{Boyd2011} does not perform well in our case, likely because the penalty for violating the constraints is reduced too heavily in some steps, leading to non-appropriate root finding starting values and followingly, no convergence to a feasible root.

Another aspect of the choice of $\rho$ that needs to be considered is \Pat{the numerical calculation of the first minimization problem of the ADMM stated in Equation~\eqref{eq:ADMMsteps} for $i = 1$. Again, it is a non-convex problem. However, the minimization subproblem is smooth and can be rephrased to a root finding problem by applying the Karush-Kuhn-Tucker-Theorem (KKT) (see Appendix \ref{subsec:app_PCAsubproblem} for more details). } The main \Pat{algorithmic} issue there is to use a good starting value \Pat{that leads to a feasible root fulfilling the inequality condition in the KKT-Theorem}. We use the \textit{warm start} approach to improve convergence speed \Pat{of the root finder}. By using the result of the prior iteration, we need fewer iterations per subproblem. By choosing $\rho$ on the larger side, our starting value~$\bb{v}_{(0)}^m$ for the root problem is also closer to the solution for~$\bb{v}_{(1)}^{m+1}$, and solving the problem becomes more stable. Specifically, we found that starting with a value of~$\eta + \frac{1}{2N}\sum_{i = 1}^N \sum_{j = 1}^p (\bb{\hat{\Sigma}}_i)_{jj}$ and increasing it sequentially if there is no convergence is working well regarding the root finding as well as the convergence of the ADMM itself. This is somehow encouraged by the findings of \citet{Ghadimi2014} regarding the optimal $\rho$ for quadratic problems, which depends on~$\eta$ and the eigenvalues of the matrix in the quadratic term. 

In order to adapt to the decreasing variability when calculating multiple components, the penalization parameter is adapted to the variance left. More specifically, this means that for calculating the $k$-th PC the penalty parameter $\rho_k$ is set to 
\begin{align*}
    \rho_k = \eta + \frac{1}{2N}\sum_{i = 1}^N \left(\sum_{j = 1}^p \left(\bb{\hat{\Sigma}}_i\right)_{jj} - \sum_{l = 1}^{k -1} (\bb{v}^l)'\bb{\hat{\Sigma}}_i\bb{v}^l\right).
\end{align*}

While high values of $\rho$ stabilize the root finding algorithm and can be used if convergence issues arise, we may need more ADMM iteration steps and convergence speed might decrease. Thus, if the algorithm does not converge within the given number of maximal steps, checking the trajectory of the residuals is helpful to differentiate between an increase of maximal steps and/or a decrease of $\rho$ if convergent behavior occurs in the residuals, or an increase in $\rho$ if non-convergent or erratic changes are observed. Again note that since we are confronted with a non-convex problem, there is no guarantee for convergence. However, we achieve convergence in all simulations and real data examples using the algorithmic fine tuning and parameter choices described above.

\subsubsection{Algorithmic enhancements}
\label{subsec:alg_enhancements}
To further improve performance and enhance computational speed, we implement several additional algorithmic enhancements.

Firstly, we apply \textit{Inexact Minimization / Early Termination} for the subproblem regarding the function $f_1(\bb{v}_{(1)})$, meaning that iterations of the root finder are stopped before full convergence during each step of the ADMM. Thus, fewer iterations per ADMM iteration are needed, leading to speed gains without significant loss of accuracy overall. The tolerance~$\epsilon_{root}$ indicates an error for finding the root of the function $f$ of $10^{-1}\epsilon_{root}|f|+10^{-1}\epsilon_{root}$ \citep[see function \texttt{multiroot} in package \texttt{rootSolve}, ][]{rootSolve} and we allow an increased error of~$10\epsilon_{root}$ in the constraints. In our calculations, $\epsilon_{root}$ is set to~$10^{-2}$ or~$10^{-1}$ for increased speed.

Secondly, the algorithm is considerably faster if we project~$\bb{v}_{(0)}^m$ after each ADMM iteration step onto the feasible subspace given by the optimization constraints in Equation~\eqref{eq:objective}. Denote the matrix containing all calculated loadings of source~$a_i$ as~$\bb{v}^{1:(k-1)}_{\cdot i} = (\bb{v}^1_{\cdot i}, \ldots, \bb{v}^{k-1}_{\cdot i})$. Then, the function projecting a vector~$\bb{v}=(\bb{v}_{\cdot 1}',\ldots ,\bb{v}_{\cdot N}')'$ to the feasible space for the $k$-th PC is defined as 
\begin{align}
\label{eq:projection}
    P_{\left(\bb{v}^{1:(k-1)}\right)^\bot}(\bb{v}) = \left(
     \begin{array}{c}
     \frac{\bb{v}_{\cdot 1} - \sum_{l=1}^{k-1} \langle \bb{v}_{ \cdot 1}, \bb{v}^{l}_{\cdot 1} \rangle \bb{v}^{l}_{\cdot 1} }{||\bb{v}_{\cdot 1}- \sum_{l=1}^{k-1} \langle \bb{v}_{ \cdot 1}, \bb{v}^{l}_{\cdot 1} \rangle \bb{v}^{l}_{\cdot 1} ||_2} \\
     \vdots \\
     \frac{\bb{v}_{\cdot N}- \sum_{l=1}^{k-1} \langle \bb{v}_{ \cdot N}, \bb{v}^{l}_{\cdot N} \rangle \bb{v}^{l}_{\cdot N}}{||\bb{v}_{\cdot N}- \sum_{l=1}^{k-1} \langle \bb{v}_{ \cdot N}, \bb{v}^{l}_{\cdot N} \rangle \bb{v}^{l}_{\cdot N}||_2}  
     \end{array}
     \right).
\end{align}
If $k = 1$, we project onto to orthogonal space of the null vector $\bb{0}$, thus we are normalizing only.

Finally, an overview of the ADMM algorithm for the $k$-th PC is summarized in Algorithm~\ref{alg:ADMM}.

\begin{algorithm}
\small
\caption{ADMM ($\bb{\hat{\Sigma}}_1, \ldots, \bb{\hat{\Sigma}}_N,\eta, \gamma, \bb{v}^{1:(k-1)}, \bb{y}_k, m_{max}, \epsilon_{ADMM}, \epsilon_{root}, \rho$)}
\label{alg:ADMM}
\begin{algorithmic}[1]
\State Initialize $\bb{v}_{(1)}, \bb{v}_{(2)}, \bb{v}_{(3)},  \bb{u}_{(1)}, \bb{u}_{(2)}, \bb{u}_{(3)} \leftarrow  \bb{0} \in \mathbb{R}^{Np}$, $\bb{v}_{(0)}^{new} \leftarrow \bb{y}_k$, $m \leftarrow 1$
\While{$m \le m_{max}$}
    \State $\bb{v}_{(0)}^{old} \leftarrow \bb{v}_{(0)}^{new}$
    \State Solve subproblems: \Comment{See also Appendix~\ref{sec:app_subproblems}}
    \State $\bb{v}_{(1)} \leftarrow$ solution of Equation~\eqref{eq:ADMMsteps} for $i = 1$ with  $\bb{u}_{(1)}$ and $\bb{v}_{(0)}^{new}$
    \State $\bb{v}_{(2)} \leftarrow  S(\bb{v}_{(0)}^{new} -\frac{1}{\rho}\bb{u}_{(2)}, \frac{\eta\gamma}{\rho})$
    \State $\bb{v}_{(3)} \leftarrow S_G(\bb{v}_{(0)}^{new} - \frac{1}{\rho}\bb{u}_{(3)}, \frac{\eta(1-\gamma) \sqrt{N}}{\rho})$
    \State $\bb{v}_{(0)}^{new} \leftarrow \frac{1}{3} \sum_{i = 1}^3 (\bb{v}_{(i)} + \frac{1}{\rho} \bb{u}_{(i)})$
    \State $\bb{v}_{(0)}^{new} \leftarrow P_{\left(\bb{v}^{1:(k-1)}\right)^\bot}(\bb{v}_{(0)}^{new})$
    \State $\bb{u}_{(i)} = \bb{u}_{(i)} + \rho(\bb{v}_{(i)}- \bb{v}_{(0)}^{new})$
    \State $r \leftarrow \sum_{i = 1}^3 ||\bb{v}_{(i)}-\bb{v}_{(0)}^{new}||_2^2$ (primal residual)
    \State $s \leftarrow 3\rho^2 ||\bb{v}_{(0)}^{new} -\bb{v}_{(0)}^{old}||_2^2 $ (dual residual)
    \State Set $\epsilon_{prime}, \epsilon_{dual}$ according to Equation~\eqref{eq:eps_primedual}
    \If{$r < \epsilon_{prime} \text{ and } s < \epsilon_{dual}$} 
        \State \textbf{break}
    \EndIf
\EndWhile
\State $\bb{v}_{(0)}^{new} \leftarrow P_{\left(\bb{v}^{1:(k-1)}\right)^\bot}(\bb{v}_{(0)}^{new})$
\State Set entries of $\bb{v}_{(0)}^{new}$ with absolute value lower than $\epsilon_{thr} = 0.005$ to 0
\State Normalize $\bb{v}_{(0)}^{new} \leftarrow P_{\left( \bb{0} \right)^\bot}(\bb{v}_{(0)}^{new})$
\State Return $\bb{v}_{(0)}^{new} $
\end{algorithmic}
\end{algorithm}

\subsection{Hyperparameter Selection} 
\label{subsec:hyperparameters}
In Section~\ref{subsec:params_sparsity} we provide criteria to select the sparsity hyperparameters, in Section~\ref{subsec:params_ssMRCD} possible parameter settings of the ssMRCD estimator are discussed and Section~\ref{subsec:Ncomponents} provides guidance on how to determine the number of components.

\subsubsection{Tuning Parameters for (Joint) Sparsity}
\label{subsec:params_sparsity}
The tuning parameters of the objective function \eqref{eq:objective}, i.e.~the parameter for \Pat{the amount of structure in sparsity induced by the distribution of the overall penalty across entrywise and groupwise penalty terms}~$\gamma$, and the parameter for the overall amount of sparsity~$\eta$, need to be selected. Regarding the  \Pat{amount of structured/global sparsity}, \citet{Simon2013}, who introduce the mixture of the two penalty types used in our sparse PCA setting, suggest to use either $\gamma = 0.05$ or $\gamma = 0.95$, depending on the expected amount of global sparsity compared to local sparsity. In \citet{Rao2016}, a similar penalty is used for classification with hyperparameter tuning via cross-validation. 

However, in an unsupervised setting like PCA, cross-validation is not feasible, and we will resort to the following optimality criteria for the first component only due to computational efficiency. Typically in sparse PCA, a good combination of sparsity (amount of zeros in the loadings) and explained variance is the goal. In contrast to standard sparse PCA, we are also looking to reduce the number of variables by setting loadings of groups to zero. Thus, we want to maximize the explained variance $\mathcal{V}(\bb{v}) = \bb{v}'\bb{\hat{\Sigma}}\bb{v}$ and the mean of the standardized entry- and groupwise sparsity
\begin{align*}
    \mathcal{S}(\bb{v}) = \frac{1}{2} \left( \frac{\#\{v_{ji} = 0, i = 1, \ldots, N, j = 1, \ldots, p \}}{N(p-1)} + \frac{\# \{ || \bb{v}_{j\cdot}||_2 = 0, j = 1, \ldots, p\}}{p-1} \right).
\end{align*}
Since the two variables vary also over~$\eta$, the optimal~$\gamma$ is chosen to be the maximizer of the area under the curve (AUC) of sparsity~$\mathcal{S}(\bb{v})$ and the explained variance, standardized to the two extreme solutions
\begin{align}
\label{eq:scaledvar}
    \frac{\mathcal{V}(\bb{v}) - \mathcal{V}(\bb{y}^\infty_1)}{\mathcal{V}(\bb{y}^0_1) - \mathcal{V}(\bb{y}_1^\infty)},
\end{align} 
along the trajectory path for varying~$\eta$, stopping at full sparsity. 

The second sparsity parameter~$\eta$ adjusts the number of sparse entries overall. \citet{Hubert2016} and \citet{Croux2013} both use BIC-based approaches to provide a tuning rule. While this is generally possible in our setting, it becomes more complicated with multiple covariances and the additional groupwise sparsity penalty regarding degrees of freedom. Instead, we prefer to use the simpler approach of optimizing the trade-off product (TPO) of the standardized number of zero entries and the standardized explained variance \eqref{eq:scaledvar} for the first principal component,
\begin{align*}
    TPO = \left(\frac{\#\{v_{ji} = 0, i = 1, \ldots, N, j = 1, \ldots, p \}}{N(p-1)}\right) \left(\frac{\mathcal{V}(\bb{v}) - \mathcal{V}(\bb{y}^\infty_1)}{\mathcal{V}(\bb{y}^0_1) - \mathcal{V}(\bb{y}_1^\infty)} \right).
\end{align*}

For further PCs, and similar to \citet{Croux2013}, the optimal~$\eta$ is adjusted according to the residual variance to distribute sparsity more equally across the loadings. However, while \citet{Croux2013} rely on the column-wise variation of the data projected on the orthogonal space of the prior components, the covariance matrices can be exploited. For a given~$\eta$ the degree of sparsity for the $l$-th PC is~$\eta_l = g_l \eta$. To calculate~$g_l$, we use the projected covariance matrix~$\bb{\hat{\Sigma}}_i$ of the orthogonal space of the corresponding first~$l-1$ PCs, per source. Instead of a sum of (robust) univariate variability measures in the orthogonal space of~$(\bb{v}^{1:(l-1)}_{\cdot i})$, we can directly calculate the amount of multivariate variability possible to reach in the orthogonal space due to the covariance matrices given. Thus, we propose to use the sum of the first eigenvalues,~$\lambda_1$, of the projected covariance matrices as scaling factor, 
\begin{align*}
    \tilde{g_l} = \sum_{i = 1} ^N \lambda_1\left(\left( \bb{I}_p - \bb{v}^{1:(l-1)}_{\cdot i} \left(\bb{v}^{1:(l-1)}_{\cdot i} \right)'\right)
    \bb{\hat{\Sigma}}_i
    \left( \bb{I}_p - \bb{v}^{1:(l-1)}_{\cdot i} \left(\bb{v}^{1:(l-1)}_{\cdot i} \right)'\right)\right),
\end{align*}
with~$\bb{I}_p$ being the~$p$-dimensional identity matrix. Finally, the standardized value~$g_l = \tilde{g_l}/\tilde{g_1}$ is used for scaling.

\subsubsection{Hyperparameters for the ssMRCD}
\label{subsec:params_ssMRCD}

The ssMRCD plug-in estimator requires additional hyperparameters that need to be set \citep[see][]{Puchhammer2023, ssMRCD_Cran}. Since the partition into multiple sources and the weights $\bb{W}$ between them are very data dependent, there is not a general rule how to set them, except the notions that are elaborated on in Section~\ref{subsec:ssMRCD}. For the smoothing parameter $\lambda$ there is a possible selection criterion described in \citet{Puchhammer2023}. However, since it is based on local outlier detection this is not applicable in the more general case, discussed here. 

Therefore, we derive a new approach to set $\lambda$ in a more general setting based on the idea that data should be described as well as possible by the means and covariances that are produced by the ssMRCD model. The model residuals $\bb{r}_\iota$, $\iota = 1, \ldots, n$, are
\begin{align*}
    \bb{r}_\iota = \bb{\hat{\Sigma}}_{a(\iota)}^{-1/2} (\bb{x}_\iota - \bb{\hat{\mu}}_{a(\iota)}) 
\end{align*}
where $a(\iota)$ is the source index corresponding to observation $\bb{x}_\iota$ and if we have a good estimation of the data, the mean of the smallest $\alpha$-fraction of residual norms
\begin{align}
\label{eq:optimalsmoothing}
    R = \frac{1}{h}\sum_{\iota = 1}^h||\bb{r}_{(\iota)}||_2
\end{align}
should be small. Hence, minimizing $R$ will be the criterion for the optimal~$\lambda$. If the partition cannot be derived from the data context, the same approach can be used to find a good grouping or even good weights, although computationally very expensive. 

\subsubsection{Number of Principal Components}
\label{subsec:Ncomponents}

Finally, the number of principal components necessary to describe the data appropriately needs to be selected. While in classical PCA the scree plot is often used, we propose to use the cumulative percent variation (CPV) as suggested in \citet{Hubert2016}. Thus, the number of components should at least cover $80\%$ of the overall (global) variation,
\begin{align}
\label{eq:CPV}
    \frac{\sum_{l = 1}^k (\bb{v}^l)' \bb{\hat{\Sigma}} \bb{v}^l}{\text{trace}(\bb{\hat{\Sigma}})} \ge 80 \%.
\end{align}
Since \citet{Hubert2016} calculate the CPV after the first robustification steps and before sparsity, the final CPV of the sparse loadings might be less than $80\%$. However, our approach uses iterative PC calculations and the availability of already robustly calculated covariance matrices. Thus, the obtained sparse loadings are ensured to explain at least $80\%$ of the variation. Depending on the research question, other summary statistics connected to CPV on a source level are applicable as well, e.g., the minimal CPV over all sources should be at least $80\%$, or an adapted scree plot consisting of boxplots can be used (see also Section~\ref{sec:realdata}).

\section{Simulations}
\label{sec:simulations}

\Pat{We introduce two simulation setups. In Section~\ref{subsec:simulationscovs} the algorithm is tested regarding the induced sparsity pattern and in Section~\ref{subsec:simulationscovs} the method and its estimated loadings are examined when applied to contaminated data.}

\begin{figure}
    \centering
    \includegraphics[width = \textwidth]{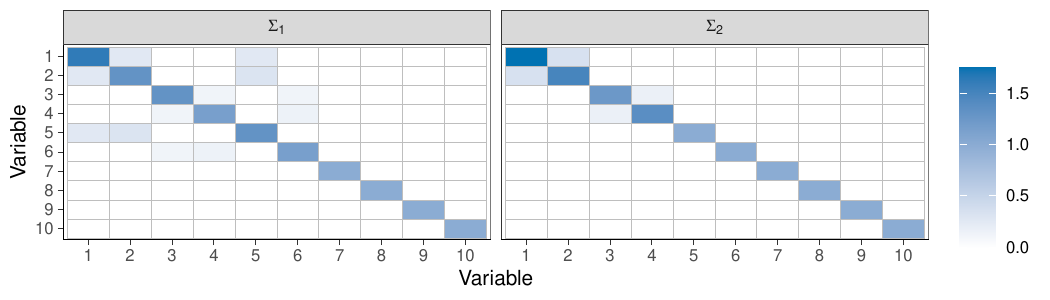}
    \caption{Heat map of the two covariance matrices used as basis for all simulation settings for $p = 10$. Each covariance entry is colored according to its value.}
    \label{fig:simulsetting}
\end{figure}

\subsection{Detecting Sparsity Patterns}
\label{subsec:simulationscovs}
The first simulation setup \textit{perturbed covariance matrices} is used to identify if the algorithm works well and if the objective function results in entry- and groupwise sparse loadings as suggested. Two covariance matrices for two groups are constructed using different sparse loading matrices $\bb{P}_1$ and $\bb{P}_2$, similar to the simulation setting of \citet{Croux2013},
\begin{align}
\label{eq:realloadings}
    \bb{P}_1 = 
    \setlength{\arraycolsep}{1pt} 
        \begin{bmatrix}
        \sqrt{\nicefrac{1}{2}} & 0 & -\sqrt{\nicefrac{1}{2}} & 0 & 0 & 0 &  \\
        \sqrt{\nicefrac{1}{4}} & 0 & \sqrt{\nicefrac{1}{4}} & 0 & -\sqrt{\nicefrac{1}{2}} & 0 &  \\
        0 & \sqrt{\nicefrac{1}{2}} & 0 & -\sqrt{\nicefrac{1}{2}} & 0 & 0 &  \\
        0 & \sqrt{\nicefrac{1}{4}} & 0 &\sqrt{\nicefrac{1}{4}} & 0 & -\sqrt{\nicefrac{1}{2}} &  \bb{0}_{6 \times (p-6)} \\
        \sqrt{\nicefrac{1}{4}} & 0 & \sqrt{\nicefrac{1}{4}} & 0 & \sqrt{\nicefrac{1}{2}} & 0 &   \\
        0 & \sqrt{\nicefrac{1}{4}} & 0 & \sqrt{\nicefrac{1}{4}} & 0 & \sqrt{\nicefrac{1}{2}} &   \\
          \multicolumn{6}{c}{\bb{0}_{(p-6) \times 6}}  &  \bb{I}_{p-6} 
        \end{bmatrix}, \
    \bb{P}_2 = 
    \setlength{\arraycolsep}{1pt} 
    \begin{bmatrix}
        \sqrt{\nicefrac{2}{3}} & 0 & -\sqrt{\nicefrac{1}{3}} & 0  &  \\
        \sqrt{\nicefrac{1}{3}} & 0 & \sqrt{\nicefrac{2}{3}} & 0 &  \\
        0 & \sqrt{\nicefrac{1}{3}} & 0 & -\sqrt{\nicefrac{2}{3}}  &  \bb{0}_{4 \times (p-4)} \\
        0 & \sqrt{\nicefrac{2}{3}} & 0 & \sqrt{\nicefrac{1}{3}}  &  \\
          \multicolumn{4}{c}{\bb{0}_{(p-4) \times 4}} &  \bb{I}_{p-4}
        \end{bmatrix},
\end{align}
and a common eigenvalue diagonal matrix,~$\bb{D} = \mbox{Diag}(2, 1.5, 1.25, 1.125, 1, \ldots, 1)$. 
The covariance matrices are then constructed based on the eigen-decomposition for each source as~$\bb{\Sigma}_i = \bb{P}_i \bb{D} \bb{P}_i'$, for~$i=1,2$ \Pat{and visualized in Figure~\ref{fig:simulsetting}}. Random noise~$\epsilon \sim \mathcal{N}(0, 0.1)$ is (symmetrically) added for each entry of the covariance matrices for each simulation run \Pat{individually to address possible uncertainty in the covariance estimation of $\bb{\hat{\Sigma}}_1$ and $\bb{\hat{\Sigma}}_2$}. While the first loadings are directionally similar, the second loadings imply opposing directions of highest variance between different sources to cover also the scenario of non-compliant dominating groups. We consider 100 repetitions for~$p = 10, \epsilon_{root} = 0.1, \eta = 0, 0.05, \ldots, 1.25, \gamma = 0, 0.5, 1$, and the first two principal components. 

\Pat{The resulting loading entries for varying~$\eta$ and~$\gamma$ are visualized in Figure~\ref{fig:entries1} for the first component and in Figure~\ref{fig:entries2} for the second component. Colored and non-solid lines indicate the variables whose loading entries should not be zero, and the corresponding horizontal lines the respective values, according to the true loadings in Equation~\eqref{eq:realloadings}. The gray variables are loading entries that should be sparse, and the shaded area around each loading entry indicates the standard error. We can clearly see that structured sparsity patterns are recovered successfully,  especially in the more complex second component that also show more deviation over the 100 simulations. While the loadings for the first PC are not indicating a significant effect of $\gamma$, in the second PC we see large differences in the variables 3 and 4, that are not sparse loading entries in both sources. By increasing~$\gamma$ they are kept at more accurately high levels for a larger range of~$\eta$ until the rise of the gray solid lines around~$\eta \ge 0.8$ indicates a trickling down of the variability of the real first PC, that is not accounted for in the estimated fully sparse first PC for high $\eta$. Altogether we can show the possibility of providing better estimates with structured sparsity patterns in appropriate settings.}

\begin{figure}
    \centering
    \includegraphics[width = \textwidth, trim={0cm 0.25cm 0 0.25cm}, clip]{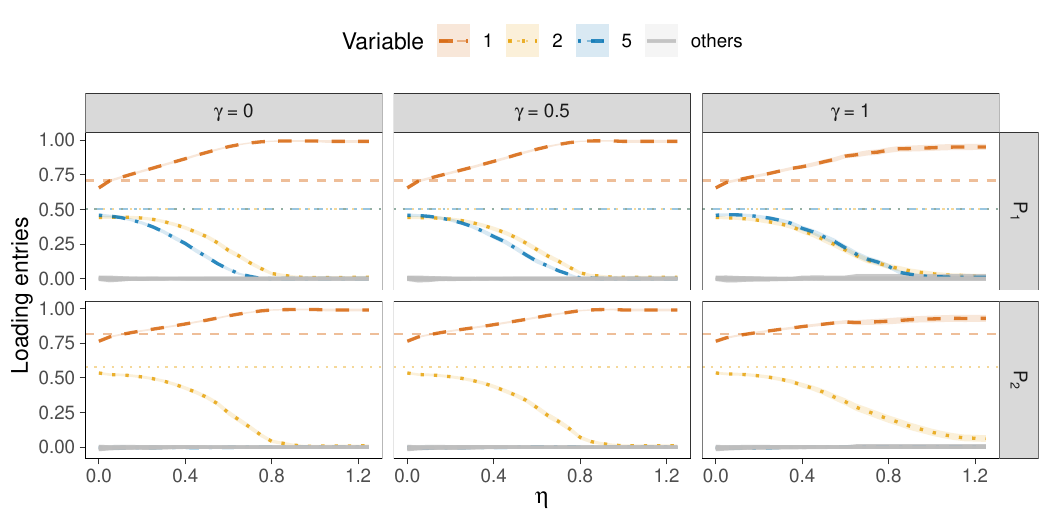}
    \caption{The mean of loading entries of the first PC with a band for the standard error. The \Pat{seven} solid gray lines depict the entries that are sparse by construction, the dashed and dotted lines the variables that are not sparse in at least one source loading. \Pat{The horizontal lines indicate the true value of the corresponding loading entries.}}
    \label{fig:entries1}
\end{figure}

\begin{figure}
    \centering
    \includegraphics[width = \textwidth, trim={0cm 0.25cm 0 0.25cm}, clip]{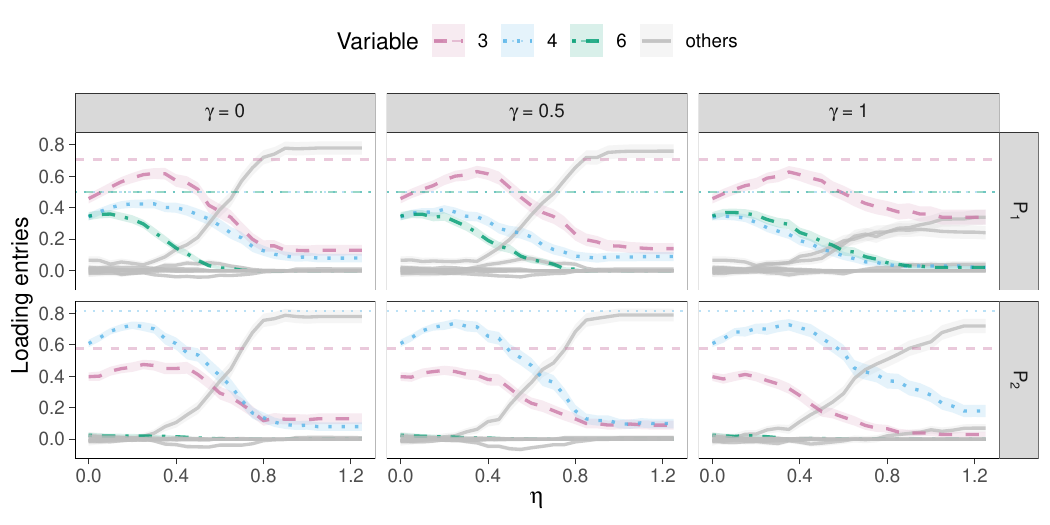}
    \caption{The mean of loading entries of the second PC with a band for the standard error. The \Pat{seven} solid gray lines depict the entries that are sparse by construction, the dashed and dotted lines the variables that are not sparse in at least one source loading. \Pat{The horizontal lines indicate the true value of the corresponding loading entries.}}
    \label{fig:entries2}
\end{figure}

\Pat{Further, we refer to the Appendix \ref{subsec:app_startingvalues_simul} for an analysis on the suitability and efficiency of the proposed starting values.}

\subsection{Outlier Robustness}
\label{subsec:simulationsdata}

The second simulation setting is used to compare different existing PCA approaches with the proposed combination of the ssMRCD estimator and groupwise sparse PCA. For a sensible application scenario of the ssMRCD estimator, we construct linearly shifting covariance matrices by using a convex combination of the two covariance matrices $\bb{\Sigma}_1$ and $\bb{\Sigma}_2$ of the prior simulation setting. With $N \ge 2$ being the number of covariances, the covariance $\tilde{\bb{\Sigma}}_i$ for source $a_i$ for $i = 1, \ldots, N$ is constructed according to
\begin{align*}
    \tilde{\bb{\Sigma}}_i =   (1-\frac{i-1}{N-1}) \bb{\Sigma}_1 + \frac{i-1}{N-1}\bb{\Sigma}_2.
\end{align*}
The corresponding real loadings are then just the eigenvectors of $\tilde{\bb{\Sigma}}_i$.

Similar to the simulation setting of \citet{Croux2013}, clean data points for each source~$a_i$ are drawn from a multivariate normal distribution~$\mathcal{N} (\bb{0}, \tilde{\bb{\Sigma}}_i)$ and a certain~$\epsilon_{out}$ fraction of shift outliers are drawn from~$\mathcal{N} (\bb{\mu}_{out}, \bb{I}_p)$ with 
\begin{align*}
    \bb{\mu}_{out} = \sqrt{2}(2,4,2,4,0,-1,1,0, 1, -1, \ldots, 0,1,-1)'
\end{align*}
per source. 

For the ssMRCD-based methods, the covariance estimation uses \Pat{either~$\alpha = 0.5$ for the robust or $\alpha = 1$ for the non-robust setting}, a band matrix, constructed with ones on the off-diagonals, zeros in the diagonal and appropriately scaled, as weight matrix~$\bb{W}$, and \Pat{either $\lambda = 0$ for the non-smoothed ssMRCD setting or the optimal smoothing criteria from Equation~\eqref{eq:optimalsmoothing} for the selection of~$\lambda$ for the other variants}. \Pat{For variants with sparsity, }the sparsity parameters are optimized using the optimization approaches described in Section~\ref{subsec:hyperparameters} on a grid of~$\gamma = 0, 0.1, \ldots, 1$ and~$\eta =0, 0.1, \ldots, 5$, or stopped prior if full sparsity is achieved. For the ROSPCA method introduced by \citet{Hubert2016}, the R-package \texttt{rospca} \citep{rospca} and the implemented optimal sparsity approach is used and applied for each source individually. 

The methods are compared using multiple evaluation criteria that are averaged over repetitions and sources. Firstly, the angle between the real subspace and the estimated subspace spanned by the first (for $k = 1$) or first and second (for $k = 2$) real and estimated loading, respectively, is calculated according to \citet{Hubert2016, Hubert2005} and standardized to $[0, 1]$. Secondly, the orthogonal distance of a non-contaminated observation $\bb{x}_\iota$,
\begin{align*}
    OD_\iota = ||\bb{x}_\iota - \bb{\hat{\mu}}_{a(\iota)} - \left(\bb{v}_{\cdot a(\iota)}^1, \ldots, \bb{v}_{\cdot a(\iota)}^k \right) \bb{t}_\iota ||_2,
\end{align*}
is given as measure for good data projection and reconstruction abilities of the components, linearly scaled to $[0,1]$ for illustration purposes. \Pat{As reference we also provide the amount of chosen sparsity by each method
\begin{align*}
    \frac{\#\{v_{ji} = 0, i = 1, \ldots, N, j = 1, \ldots, p \}}{N(p-1)}
\end{align*}}

From a sparsity recognition point of view, standard evaluation techniques can be analyzed to check if methods correctly specify sparse and non-sparse variables. The \textit{True Negative Rate} (TNR) specifies the percentage of correctly identified non-sparse entries of the loadings
and the \textit{True Positive Rate} (TPR) the correctly found sparse entries,
\begin{align*}
    TNR = \frac{1}{N} \sum_{i = 1}^N  \frac{\#\{j \in {1, \ldots, p}: \hat{v}_{ji} \neq 0, v_{ji} \neq 0\}}{\#\{j \in {1, \ldots, p}: v_{ji} \neq 0\}}, \\
    TPR = \frac{1}{N} \sum_{i = 1}^N  \frac{\#\{j \in {1, \ldots, p}: \hat{v}_{ji} = 0, v_{ji} = 0\}}{\#\{j \in {1, \ldots, p}: v_{ji} = 0\}}.
\end{align*}

We also include three evaluation measures to combine the two measures into one. The well known F1-Score and the zero-measure (Z-measure) introduced in \citet{Hubert2016}, which calculates the percentage of overall correctly identified entries without partitioning into groups first, are applicable in a balanced setting. However, in an unbalanced setting with high sparsity, F1 and the Z-measure essentially lead to ignorance of correctly identified non-sparse entries and the amount of correctly identified sparse entries drives a "good" performance and gives incentives to overestimate the sparsity pattern. Thus, we prefer to use the geometric mean (G-Mean) of TNR and TPR. 

\begin{figure}
    \centering
    \includegraphics[width = \textwidth, trim={0cm 0.65cm 0 0.3cm}, clip]{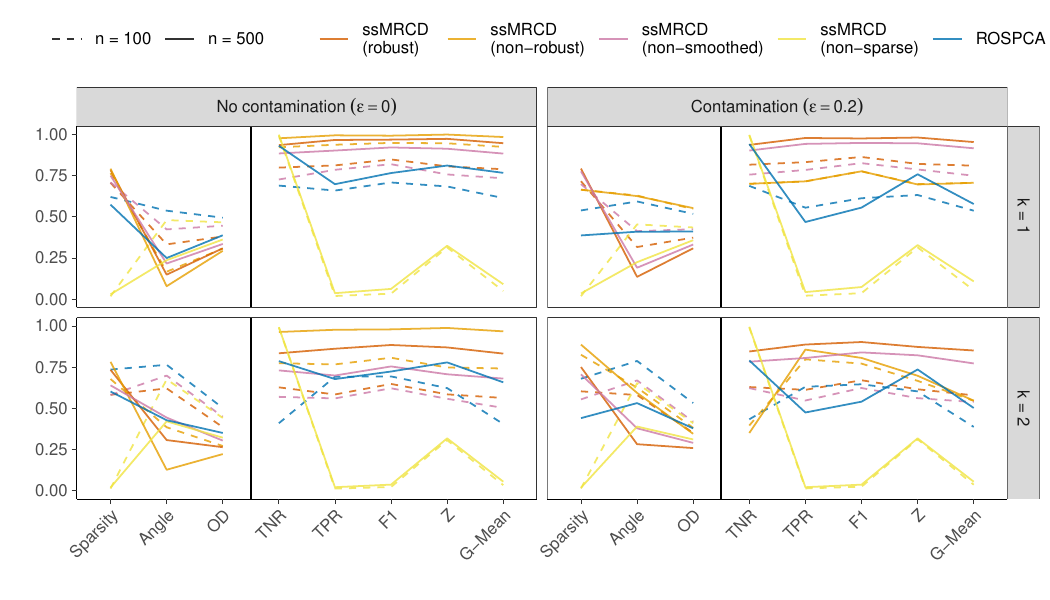}
    \caption{Comparison of performance of the first two components for simulation scenario 2 with equally contaminated data sources ($N = 10$). Solid lines are the mean performance for $n = 500$ observations, dashed lines for $n = 100$ observations.}
    \label{fig:simul2_result}
\end{figure}

In Figure~\ref{fig:simul2_result} our proposed method (ssMRCD) with and without robustness, smoothing or sparsity, is compared to ROSPCA. For each method we calculate the first two PCs for 100 repetitions and $p = 10$ with varying contamination level $\epsilon = 0\%, 20\%$ and number of data observations $n = 100, 500$, both constant over all $N = 10$ sources. We can see that for the first component the robust ssMRCD-based method provides better results in all measurements and settings than ROSPCA. Especially interesting compared to ROSPCA is the combination of higher sparsity with a lower angle and low OD. This implies a good \Pat{fit of the highly sparse loadings to the data and the subspace of highest variation}. Moreover, the good measurement scores for classification confirm that the correct sparsity structure is found. Compared to the non-robust of the ssMRCD-based method, the robust version is slightly less effective, especially with a lower number of observation. However, contamination is heavily affecting the non-robust method. Criteria connected to data reconstruction, i.e. the angle and OD, show inferior performance while the detection of sparsity patterns is comparable to ROSPCA. \Pat{When using the ssMRCD estimator without smoothing, the performance worsens slightly over all scenarios and performance measures. Lastly, the plug-in ssMRCD PCA method without sparsity is evaluated, showing clear and expected disadvantages compared to all other methods.}

When it comes to the second PC, ssMRCD methods (non-robust for $\epsilon = 0\%$, robust for $\epsilon = 20\%$) are still the best performing for $n = 500$. However, for a reduced number of observations $n = 100$ the performance of classifying the correct sparsity pattern is comparable or just slightly better than ROSPCA. Interestingly in some settings, ROSPCA does not take advantage of an increase in sample size. \Pat{The non-smoothed and the non-sparse ssMRCD-based methods show the same patterns as for PC1.}

\begin{figure}
    \centering
    \includegraphics[width = \textwidth, trim={0cm 0.65cm 0 0.3cm}, clip]{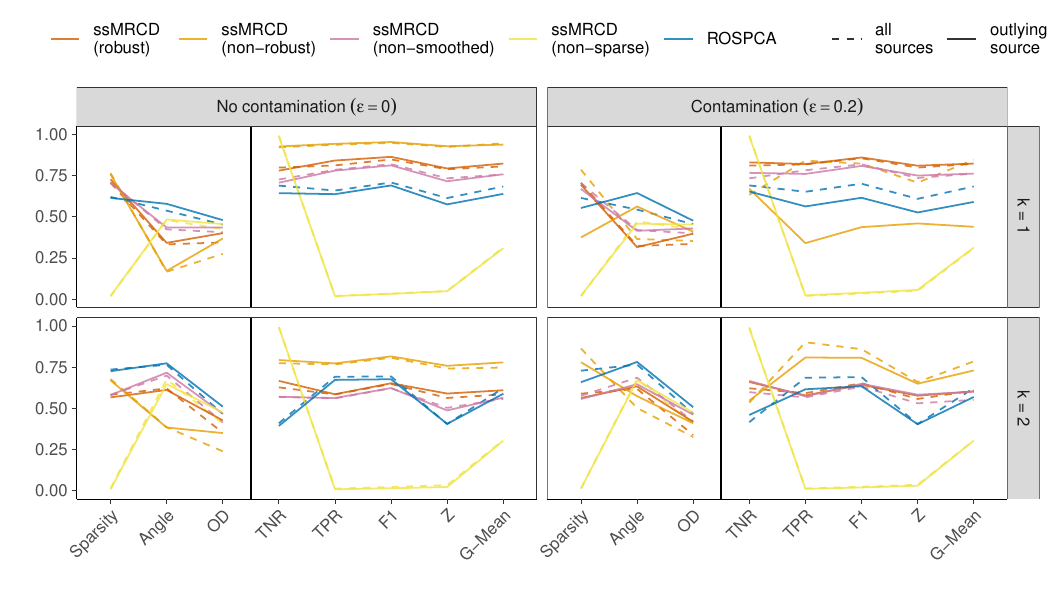}
    \caption{Comparison of performance of the first two components for simulation scenario 2 with locally contaminated data in source $5$ of $N = 10$ sources for $n = 100$. The result for data source $5$ is shown with solid lines, the mean over all sources is depicted with dashed lines.}
    \label{fig:simul3_result}
\end{figure}

Another interesting aspect to analyze is how \textit{locally contaminated data}, in this context meaning contamination in only one source, affects PCA results. Therefore, we stick to the data contamination setting with $N = 10$ data sources but instead of contaminating all sources equally, we only contaminate the fifth of $N = 10$ sources. Again, we use $n = 100$ observations, $100$ repetitions and $p = 10$ variables. The results are shown in Figure~\ref{fig:simul3_result}.

Again, we can see that the robust ssMRCD method generally provides consistently better results in PC1 and better or at least comparable results in PC2. However, when we analyze the results for the single contaminated data source, we see very stable results for the robust ssMRCD method. This is in contrast to the results of ROSPCA with contamination, where the performance on the contaminated source is clearly worse in almost all performance measures for both components. The non-robust version of the ssMRCD-based method also shows very good performance for the uncontaminated case. With local contamination we can see a stark performance decline in the contaminated source, especially in the first PC. Interestingly, the average over the sources is not affected to that extend but remains somehow stable. This indicates that provided a multi-source scenarios, inherent similarities in the covariances between groups should be leveraged. Applying additional smoothing stabilizes the covariance estimation, even in a non-robust setting, and in combination with groupwise sparsity we achieve reliable sparse loadings.

\section{Applications}
\label{sec:realdata}

We demonstrate the usefulness of the proposed sparse multi-source PCA method on two diverse applications, namely one on \Pat{multivariate time series data from an Austrian weather stations (Section~\ref{subsec:weather})}  and the second on \Pat{measurements of plant geochemistry (Section~\ref{subsec:plant})}.

\subsection{Weather Analysis at Hohe Warte}
\label{subsec:weather}

\begin{figure}
    \centering
    \includegraphics[width = \textwidth, trim={0cm 0.25cm 0 0.25cm}, clip]{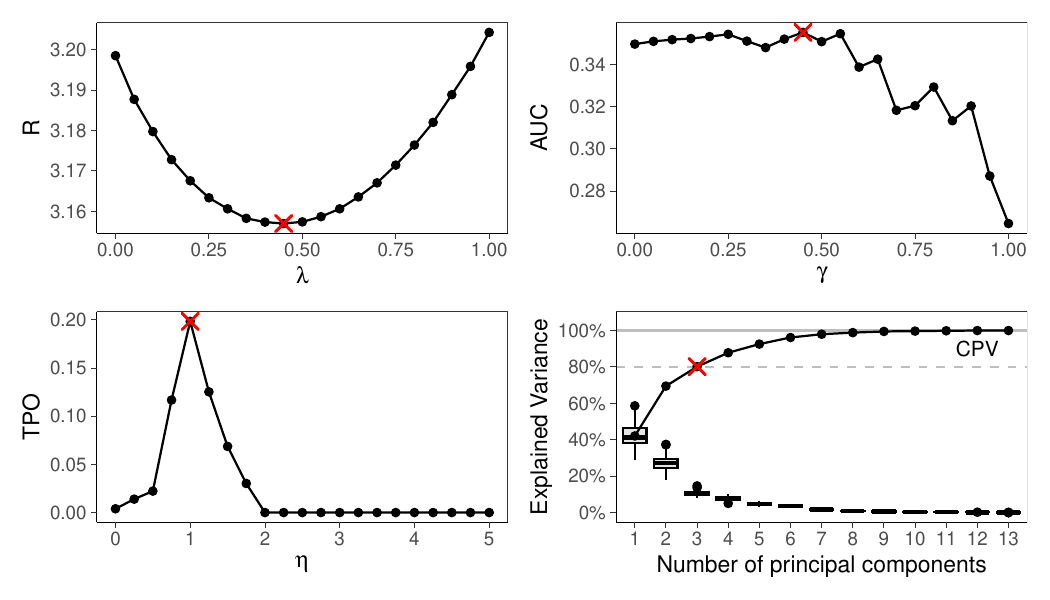}
    \caption{Hohe Warte weather station: Optimal smoothing and sparsity parameters.}
    \label{fig:weather_optparscree}
\end{figure}

We analyze daily weather measurements ($n=23372$) of the weather station \textit{Hohe Warte} in Vienna, Austria over the years 1960-2023 as provided by \citet{Geosphere2024}. The dataset consists of $p=13$ variables covering the amounts of sunshine, wind, cloud coverage as well as temperature, humidity, air pressure and visibility (see Appendix~\ref{sec:app_weathervars} for the full list) for $N=64$ sources corresponding to the different years. For preprocessing, we standardize the variables to the corresponding medians and the mean absolute deviations from \Pat{the years 1960 to 1980 that are used as a baseline for proceeding climatic developments}.

\begin{figure}
    \centering
    \includegraphics[width = \textwidth, trim={0cm 2.9cm 0 2.3cm}, clip]{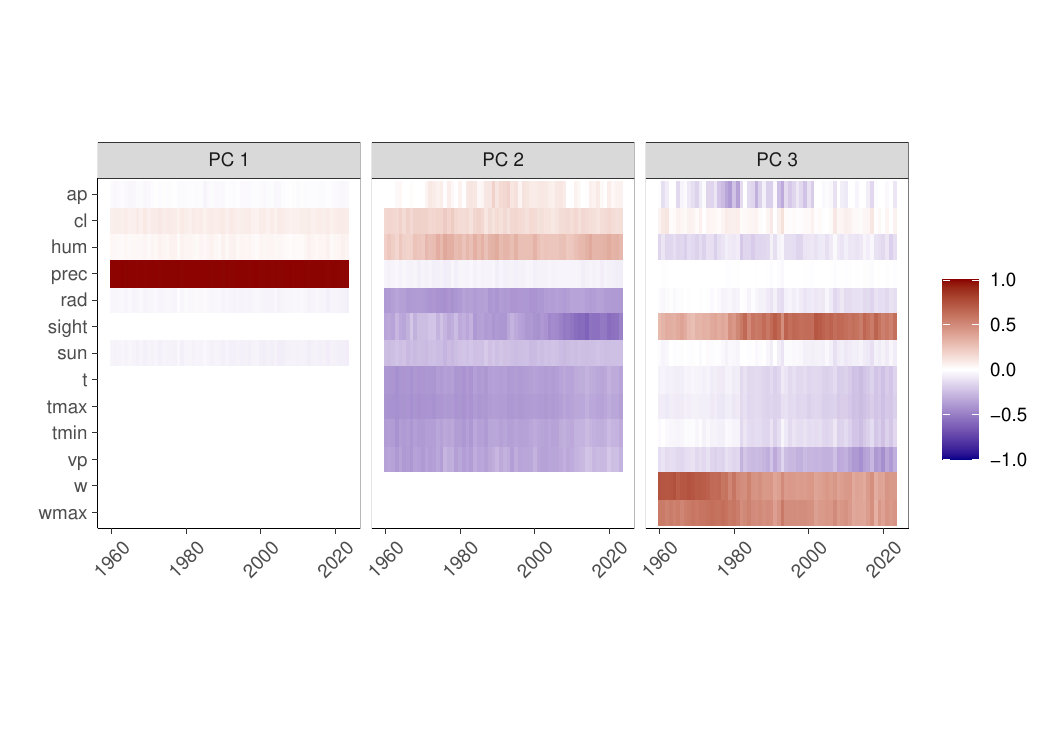}
    \caption{Hohe Warte weather station: Heat map of the loadings for the $p=13$ variables (rows) over time (columns) of the first three components (panels).}
    \label{fig:weather_loadings}
\end{figure}

For the computation of the ssMRCD plug-in estimator, we assume that each year is similar to five prior and five subsequent years with a linear decrease in similarity leading to the weighting matrix $\bb{W}$ structured as a band matrix given by
\begin{align*}
    \left( \begin{array}{ccccccc}
         0      & 5      & \ldots & 1      &         &   \\
         5      & \ddots & \ddots &        &  \ddots &   \\
        \vdots  & \ddots & \ddots & \ddots &         & 1 \\
        1       &        & \ddots & \ddots & \ddots  & \vdots \\
                & \ddots &        & \ddots & \ddots  & 5 \\
                &        &  1     & \ldots & 5       & 0
    \end{array} \right),
\end{align*}
where each row is then scaled to have an overall sum of 1 \citep[see also][]{Puchhammer2023}. The amount of smoothness,~$\lambda$, of the ssMRCD estimator is optimized over the interval~$[0,1]$ with step size~$0.05$ to minimize the residual norm in Equation~\eqref{eq:optimalsmoothing}. The upper left part of Figure~\ref{fig:weather_optparscree} shows the residual norm varying over the amount of smoothing and the optimal value of~$\lambda = 0.45$. Based on the optimally smoothed set of covariances, the sparsity parameters of the proposed sparse multi-source PCA method are selected using a step size of~$0.05$ for~$\gamma$ and~$0.25$ for~$\eta$. The optimal values are then~$\gamma = 0.25, \eta = 1$ (see upper right and lower left part of Figure~\ref{fig:weather_optparscree}, respectively). The corresponding boxplot-based scree-plot is shown in the lower left part of Figure~\ref{fig:weather_optparscree}. \Pat{Each boxplot is constructed per PC and is based on the individually explained variance per year for all $N = 64$ years.} According to the CPV-criterion Equation~\eqref{eq:CPV}, it is sufficient to analyze only the first three components, since they explain $80\%$ of the overall variance.

\begin{figure}
    \includegraphics[width = \textwidth, trim={0cm 0.2cm 0 0cm}, clip]{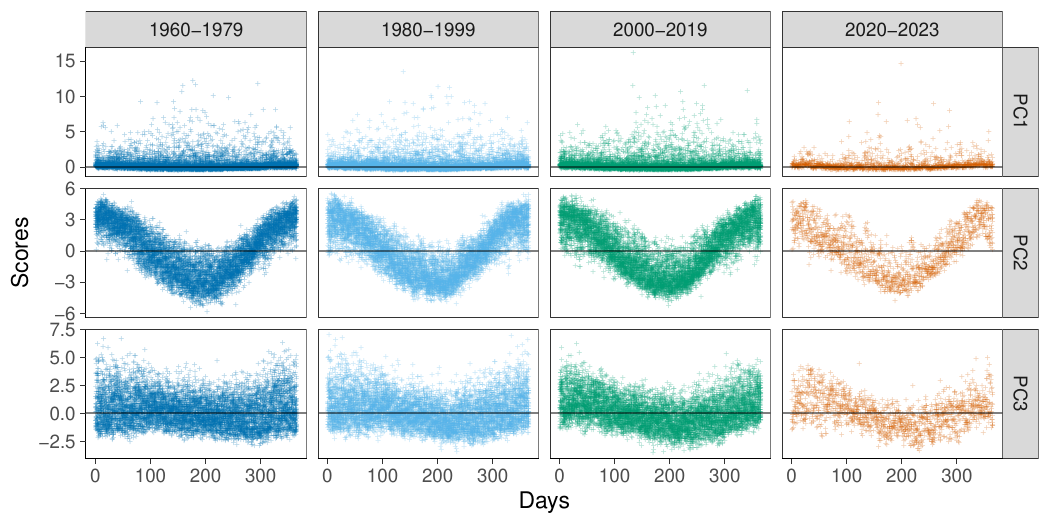}
    \caption{Hohe Warte weather station: Scores for each observation of the first three components (rows) partitioned into four consecutive time subsets (columns).}
    \label{fig:weather_scores}
\end{figure}

Figure~\ref{fig:weather_loadings} presents the sparsity patterns of the loadings for the first three PCs obtained by our sparse multi-source PCA method. We see three different causes of variation. The first PC (left panel) is mainly composed of precipitation and displays a clear global pattern across all years, meaning that precipitation drives most of the weather variation over the year. Figure~\ref{fig:weather_scores} displays the corresponding scores (top row), partitioned into time subsets, and shows rather constant variation over time. 
The second component (middle panel in Figure~\ref{fig:weather_loadings}) consists mainly of temperature, vapor pressure, sight, radiation and sun as well as humidity and cloud cover in the opposed direction. By inspecting the scores in the middle row of Figure~\ref{fig:weather_scores}, we can see that the second component captures the seasonality and the corresponding variability over each year. Again, the pattern seems to be rather stable over the 64 years. 

Finally, the loadings for the third component are visualized in the right panel of Figure~\ref{fig:weather_loadings} and consist again of temperature variables and sight. Yet, in contrast to the first two PCs, a trend becomes visible in their loadings on the third PC as can be seen from darker colors in the heat map for the more recent years. Also maximal wind speed and wind velocity display important loading entries, which are stable or rather decreasing in importance for variability over the years, respectively. This changing pattern over the years is also visible in the shape of the scores displayed in Figure~\ref{fig:weather_scores}, bottom row. While the years 1960-1979 do not seem to exhibit seasonality in the scores, such a pattern becomes more apparent for the more recent years. While pinpointing the source of variability for the third component is more difficult than for the first two components, one possible source could be related to climate change, 
apparent from the evolving trend over the long time frame of 64 years. The smooth transitioning over time as mainly visible in the third component is directly detectable from our multi-source PCA method whereas it would remain unnoticed from a standard (global) PCA analysis.

\subsection{Geochemical Plant Analysis}
\label{subsec:plant}

Our second application demonstrates the usefulness of sparse  PCA for multi-source data with a more general grouping structure.The data consists of $n=547$ observations of $p=19$ element concentrations originating of $N=6$ different plant species (Norway Spruce, Common Juniper and Scots Pine) and organs (bark, needle, twig) and was collected during the NEXT project funded by the EU \citep{NEXTdata} in order to draw conclusions for mineral exploration. The aim is to explore differences and similarities of variance among the different plant groups and the suitability of the scores to discriminate between mineralizations and non-mineralizations.

For applying the ssMRCD estimator, we assume equal amounts of similarity between observations of the same plant species or of the same organs when constructing the weight matrix $\bb{W}$. Moreover, due to the compositional nature of element concentrations, we apply the standard \textit{isometric log-ratio}(ilr) transformation to the data \citep[see e.g.][]{Filzmoser2018} known from compositional data analysis. Based on the optimal smoothing criteria $R$, the optimal value for group smoothing is $\lambda = 0.35$.

\Pat{After the calculation of the ssMRCD covariance matrices, they are (linearly) transformed from ilr to \textit{centered log-ratios} (clr) coordinates to increase the interpretability of the principal components' loadings and scores. The clr-transformation essentially standardizes each variable with the geometric mean per observation, followed by a log-transformation.} While clr leads to linear dependent variables opposed to ilr, leading to numerical issues for covariance estimation (especially determinant based estimators like the ssMRCD estimator), it is possible to intuitively interpret clr as relative importance of elements which is not possible for ilr. Thus, we apply the sparse multi-source PCA algorithm to the transformed covariance matrix of clr variables. Optimal parameters are then given by $\gamma = 0.6, \eta = 0.15$ \Pat{(see also Figure~\ref{fig:plants_optpars} in Appendix~\ref{sec:app_plants}).}

In Figure~\ref{fig:plant_loadings} the loadings of the first and second principal components explaining around $33\%$ of overall variance are shown per source, being a combination of a plant species and organ\footnote{\Pat{Since 10 PCs are needed to explain $80\%$ of the data (see also Figure~\ref{fig:plants_optpars}) we will focus on the first two components for interpretation.}}. We see clear similarities across all organs of the juniper species in both components and of the spruce species for the first component. Only in the second component the organs of the same species (spruce) start to show differences. Moreover, pine bark has the most complexity in the loading structure of the first component. This combination of heavy metals like uranium (U), vanadium (V) and lead (Pb) against phosphorus (P), potassium (K) and rubidium (Rb) is to be expected from physiological characteristics of pine bark.

\begin{figure}
    \centering
    \includegraphics[width = \textwidth,
    trim={0cm 0.5cm 0 0cm}, clip]{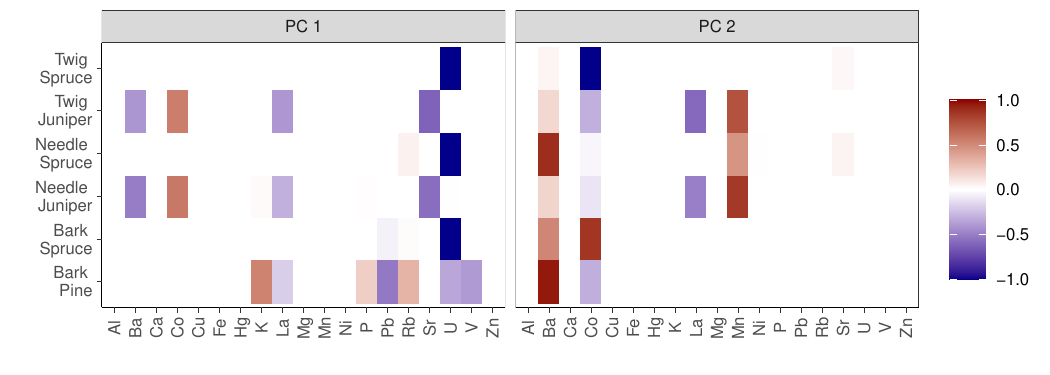}
    \caption{Plant geochemistry: Heat map of loadings for the $p=19$ variables (columns) over plant-organ species combination (rows) on the first and second principal component (panels)}.
    \label{fig:plant_loadings}
\end{figure}

Note that the loadings of standard PCA on clr coordinates sum up to zero due to the linear dependence of clr variables \citep{Jolliffe2002}. \Pat{For sparse PCA, the dependent structure of compositional data can be taken care of explicitly e.g. via sparse principal balances \citep{Mert2015}.} However, especially in our data example with only few variables, including the zero sum constraint could lead to high importance of a variable just to fulfill the constraint (e.g.\ for PC1 in group Twig Spruce in Figure~\ref{fig:plant_loadings}). We believe that by forcing the zero sum constraint in a sparse PCA setting with a small number of compositional variables, unwanted noise might be introduced, complicating the interpretation of the loadings.

When it comes to mineral exploration, a goal of the NEXT project, it would be interesting if we can find a plant organ-species combination and a direction of variation along which the discrimination between mineralizations and non-mineralizations is visible. To investigate this, we use the geological classification between calcsilicate rocks and mafic rocks. Mafic rocks are often associated with volcanic and intrusive activities and they can indicate the presence of specific mineral deposits like nickel, copper, and platinum group elements. 

Taking a look at Figure~\ref{fig:plant_score} we can see the distribution of the scores connected to the first (left) and second (right) PC as density and the median as vertical lines for different groups split into observations connected to calcsilicate and mafic rocks. Other geological units are not shown. 
The bark of Scots pine has the most differentiable peaks and medians, indicating the possibility to use this plant organ-species combination with the elements of the first loading for mineral exploration. Similar conclusions can be made for the second PC. \Pat{Here, Norway spruce tends to differentiate the most between the two geologies across all organs, indicating potential leverage for geological and mineral exploration.} However, the geology and other external variables of the respective dataset can vary heavily and other sources of variation like soil moisture, amount of till, fine fraction of the sample or physiological effects of the plants mentioned before can also be part of variation described by the PCs.

\begin{figure}
    \centering
    \includegraphics[width = \textwidth, trim={0cm 0.4cm 0 0.4cm}, clip]{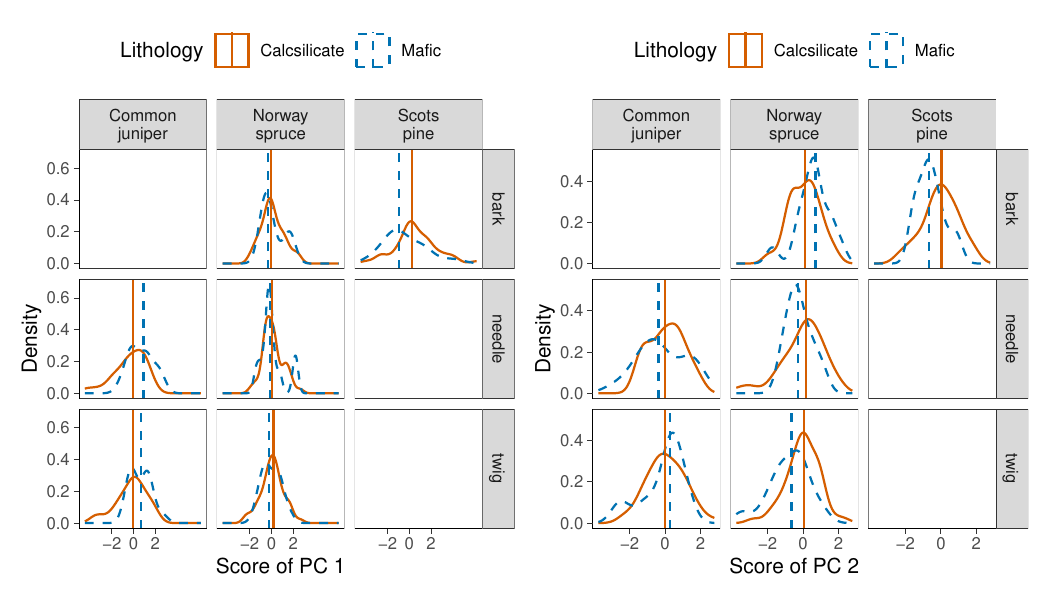}
    \caption{Plant geochemistry: Density and median (vertical line) for calcsilicate rock (solid) and mafic rock (dashed) measurements for the first PC for all plant species (columns) - organ (rows) combinations. Note that empty panels correspond to combinations of plant species and organs not present in the considered dataset.}
    \label{fig:plant_score}
\end{figure}

\section{Conclusions}
\label{sec:conclusion}

We introduce sparse PCA analysis for multiple related data sources to permit the detection of global as well as local, source-specific sparsity patterns in the PCA loadings. To this end, we propose an optimization problem that maximizes explained variances across the multiple data sources while inducing structured sparsity patterns. The ssMRCD estimator is used as plug-in into the optimization problem and perfectly fits the spirit of combined global-local sparsity patterns by delivering local covariances that are smoothed over groups. Moreover, it provides protection against the presence of outliers in the data.

We provide a computationally efficient algorithm based on the ADMM to obtain sparse outlier-robust PCA loadings. Algorithmic parameters are fine-tuned and convergence is achieved in all applications and simulations. Care is given to optimally select the hyperparameters controlling the degree of sparsity and smoothing properties of the ssMRCD estimator tailored to the PCA context.

The proposed sparse multi-source PCA method performs well in simulation settings mimicking structured sparsity and it outperforms non-robust counterparts as well as the state-of-the-art sparse, robust PCA method ROSPCA when outliers are present. \Pat{The versatility of the multi-source method is illustrated on two different applications.}

\Ines{
\Pat{Possible further application scenarios entail also a wide variety of data, where the grouping structure is not fixed upfront. The flexibility of the method regarding the source-definition can also be leveraged for, e.g., the large field of spatial data.}
Finally, our multi-source perspective to sparse, outlier-robust PCA holds also promise for other popular multivariate analyses such as discriminant analysis, graphical modeling or canonical correlation analysis.}

\section*{Acknowledgements} 
We sincerely thank our esteemed colleague, Solveig Pospiech, for her invaluable geological and biological expertise on plants, geology and geochemistry. Her generous contributions have greatly enhanced the precision and depth of our interpretations. Co-funded by the European Union (SEMACRET, Grant Agreement no. 101057741) and UKRI (UK Research and Innovation). The authors acknowledge TU Wien Bibliothek for financial support through its Open Access Funding Programme. Ines Wilms is supported by a grant from the Dutch Research Council (NWO), research program Vidi under the grant number VI.Vidi.211.032.

\pagebreak

\bibliographystyle{apalike}
\bibliography{references}

\pagebreak

\appendix
\setcounter{proposition}{0}

\section{ADMM minimization problems}
\label{sec:app_subproblems}

For ease of notation, we introduce the matrix notation of the vectorized ADMM components $\bb{u}_{(1)}, \bb{u}_{(2)}, \bb{u}_{(3)}$ and $\bb{v}_{(0)}, \bb{v}_{(1)}, \bb{v}_{(2)}$ and $\bb{v}_{(3)}$ similar to \eqref{eq:notation_matrix} as
\begin{align*}
     \bb{U}^m_{(i)} = \left( 
    \begin{array}{ccc}
         \bb{u}_{(i), 1}^m&  \ldots & \bb{u}_{(i), (N-1)p +1 }^m\\
         \vdots &  & \vdots \\
         \bb{u}_{(i), p}^m & \ldots & \bb{u}_{(i), Np}^m
    \end{array}
    \right) = (\bb{u}_{(i),\cdot 1}^m, \ldots, \bb{u}_{(i),\cdot N}^m) = (\bb{u}_{(i),1\cdot}^{m'}, \ldots, \bb{u}_{(i),p\cdot}^{m'})', \\
     \bb{V}^m_{(i)} = \left( 
    \begin{array}{ccc}
         \bb{v}_{(i), 1}^m&  \ldots & \bb{v}_{(i), (N-1)p +1 }^m\\
         \vdots &  & \vdots \\
         \bb{v}_{(i), p}^m & \ldots & \bb{v}_{(i), Np}^m
    \end{array}
    \right) = (\bb{v}_{(i),\cdot 1}^m, \ldots, \bb{v}_{(i),\cdot N}^m) = (\bb{v}_{(i),1\cdot}^{m'}, \ldots, \bb{v}_{(i),p\cdot}^{m'})'.
\end{align*}
The notation for the variables without superscript is likewise as well as for the vector $\bb{z}$ used in Equation~\eqref{eq:z}.

\subsection{Minimization problem for $\boldsymbol{v}_{(1)}$}
\label{subsec:app_PCAsubproblem}

Due to the block-diagonal structure of $\bb{\hat{\Sigma}}$, the additivity of the quadratic Frobenius norm and the separable constraints, the minimization problem can be separated among sources. Thus, per source $a_i$, we have the following minimization problem in  $\bb{v} \in \mathbb{R}^p$ for each iteration step $m$
\begin{align*}
    \min_{v} \hspace{0.5cm} & - \bb{v}'\bb{\hat{\Sigma}}_i \bb{v}  + \frac{\rho}{2} ||\bb{v} + \underbrace{\frac{1}{\rho}\bb{u}_{(1), \cdot i}^{m} - \bb{v}_{(0), \cdot i}^m}_{=:\bb{c}}||_2^2 \\
    s.t. \hspace{0.5cm} & \bb{v}'\bb{v} = 1,  \\
    & \bb{v}'\bb{v}_{\cdot i}^l = 0,  \hspace{0.5cm} 1 \leq l < k  \\
    & \bb{z}_{\cdot i}'\bb{v} \ge 0.
\end{align*}
The problem is non-convex but differentiable and, thus, can be solved by calculating the Lagrangian
\begin{align*}
    \mathcal{L}(\bb{v})  &= -\bb{v}'\bb{\hat{\Sigma}}_i \bb{v} + 
    \frac{\rho}{2} ||\bb{c} + \bb{v}||_2^2 - 
    \mu\bb{z}_{\cdot i}'\bb{v} + 
    \lambda_0 (\bb{v}'\bb{v}-1) + 
    \sum_{l = 1}^{k-1} \lambda_l (\bb{v}'\bb{v}_{\cdot i}^l)
\end{align*}
and applying the Karush-Kuhn-Tucker (KKT) conditions,
\begin{align}
    \nabla_{\bb{v}} \mathcal{L}(\bb{v}) &= -2\bb{\hat{\Sigma}}_i \bb{v} + \rho (\bb{c} + \bb{v}) - \mu \bb{z_{\cdot i}} + 2\lambda_0 \bb{v} + \sum_{l = 1}^{k-1} \lambda_l \bb{v}_{\cdot i}^l = \bb{0}, \label{eq:KKT1} \\
    g(\bb{v})    &= -\bb{z}_{\cdot i}'\bb{v} \le 0, \label{eq:KKT2} \\
    h_0(v)  &= \bb{v}'\bb{v} -1 = 0, \label{eq:KKT3} \\
    h_l(v)  &= \bb{v}'\bb{v}_{\cdot i}^l = 0,  \hspace{0.5cm} \forall 1 \le l < k, \label{eq:KKT4} \\
    \mu    &\ge 0,  \label{eq:KKT5} \\
    \mu \bb{z}_{\cdot i}'\bb{v}  &= 0. \label{eq:KKT6}
\end{align}
For speed we can derive a term for $\lambda_0$ by multiplying equation \eqref{eq:KKT1} from left with $\bb{v}'$, which cancels $\lambda_l$ and $\mu$ due to the optimality conditions \eqref{eq:KKT6} and \eqref{eq:KKT4}, 
\begin{align*}
 -2\bb{v}'\bb{\hat{\Sigma}}_i \bb{v} + \rho \bb{v}'(\bb{c} + \bb{v}) - \underbrace{\mu \bb{z}_{\cdot i}'\bb{v}}_{ = 0} + 2\lambda_0 \underbrace{\bb{v}'\bb{v}}_{=1} + \sum_{l = 1}^{k-1} \lambda_l \underbrace{\bb{v}'\bb{v}_{\cdot i}^l}_{=0} = 0\\ 
    \lambda_0 = \bb{v}'\bb{\hat{\Sigma}}_i \bb{v} - \frac{\rho}{2} \bb{v}'(\bb{c} + \bb{v}).
\end{align*}

It is also possible to calculate $\lambda_l$ as a function of $\mu$ and $\bb{v}$ by multiplying with $(\bb{v}^{l}_{\cdot i})'$,
\begin{align*}
    -2(\bb{v}^{l}_{\cdot i})'\bb{\hat{\Sigma}}_i \bb{v} + \rho \underbrace{(\bb{v}^{l}_{\cdot i})'(\bb{c} + \bb{v})}_{= \bar{v_j}'\bb{c}} - \mu (\bb{v}^{l}_{\cdot i})' \bb{z}_{\cdot i} + \underbrace{\sum_{l = 1}^{k-1} \lambda_l (\bb{v}^{l}_{\cdot i})'\bb{v}^{l}_{\cdot i}}_{=\lambda_l} = 0 \\
    \lambda_l = 2(\bb{v}^{l}_{\cdot i})'\bb{\hat{\Sigma}}_i \bb{v} - \rho (\bb{v}^{l}_{\cdot i})'\bb{c} + \mu (\bb{v}^{l}_{\cdot i})' \bb{z}_{\cdot i}.
\end{align*}
However, substituting $\lambda_l$ with the exact expression derived above has proven to deteriorate precision in the orthogonality constraints without a significant gain in speed.

It is not possible to derive an analytical solution due to third and higher-order terms after substituting the multiplier $\lambda_0$ into Equation \eqref{eq:KKT1}. We have to resort to solving the root constraints \eqref{eq:KKT1}, \eqref{eq:KKT4} and \eqref{eq:KKT6} numerically using the function \texttt{multiroot} from the R-package \texttt{rootSolve} \citep{rootSolve}. Additionally, we need to ensure that all other constraints are also fulfilled after finding a root.  we apply the concept of \textit{warm starts} and use $\bb{v}^m_{(0),\cdot i}$ as the starting value for the root finder. If no feasible root is found, we increase $\rho$ until a feasible root is found.  

Regarding regularity conditions, we can check the \textit{linear independence constraint qualification} (LICQ) condition, where we need linear independence of all $\nabla h_l(\bb{v}) = \bb{v}_{\cdot i}^l$, $\nabla h_0(\bb{v}) = \bb{v}$ and $\nabla g(\bb{v}) = \bb{z}_{\cdot i}$ if $g(\bb{v}) = \bb{z}_{\cdot i}'\bb{v} = 0$. By design, $\nabla h_l(\bb{v})$ and $\nabla h_0(\bb{v})$ are independent since the components are all orthogonal. If $g(\bb{v}) = \bb{z}_{\cdot i}'\bb{v} = 0$, we are orthogonal to $\nabla h_0(\bb{v})$. Additionally choosing $\bb{z}_{\cdot i}$ orthogonal to all prior loadings, the regularity condition is fulfilled for all $\bb{v}$, implying that it is sufficient to look at points fulfilling the KKT conditions to find the optimum. Since for each source $a_i$ $\bb{z}_{\cdot i}$ is chosen as the starting value $\bb{y}_{k, \cdot i}$ which is in the given feasible space and thus part of the orthogonality space of~$\bb{v}_{\cdot i}^l$, the LICQ condition is fulfilled.

\subsection{Minimization problem for $\bb{v}_{(2)}$}

The objective function to minimize, 
\begin{align*}
    \eta \gamma  ||\bb{v}_{(2)}||_1 + \frac{\rho}{2}|| \frac{1}{\rho} \bb{u}_{(2)}^m + \bb{v}_{(2)} - \bb{v}_{(0)}^m ||^2_2,
\end{align*}
is separable across sources due to the squared Frobenius norm and the L$_1$-norm. The analytical solution is thus given by the proximal operator of the L$_1$-norm, i.e. element-wise soft thresholding \citep{Boyd2011, Wilms2022} for each entry of $\bb{v}_{(2)}$,
\begin{align*}
\label{eq:min2}
    \bb{v}_{(2), i}^{m+1} = S\left( \bb{v}_{(0),i}^m -\bb{u}_{(2),i}^m/\rho, \eta\gamma/\rho\right) \hspace{0.5cm} \forall i = 1, \ldots, Np,
\end{align*}
with $S(x, \lambda) = \text{sign}(x) \max(|x| - \lambda, 0)$.

\subsection{Minimization problem for $\bb{v}_{(3)}$}

The part of the minimization function connected to the groupwise sparsity, $f_3(\bb{v}_{(3)})$, can be rewritten as
\begin{align*}
    f_3(\bb{v}_{(3)}) &= \eta (1-\gamma) \sqrt{N}\sum_{j= 1}^p \sqrt{\bb{v}_{(3)}' \bb{C}_j \bb{v}_{(3)}}= \eta (1-\gamma) \sqrt{N}\sum_{j= 1}^p ||\bb{v}_{(3), j.}||_2.
\end{align*}
Thus, we can use the groupwise/block soft thresholding operator, which is the proximal operator of the L$_1$-norm of subgroups \citep{Boyd2011, Wilms2022}
\begin{align*}
    \bb{v}_{(3), j\cdot}^{m+1} = S_G(\bb{v}_{(0), j\cdot}^m - \bb{u}_{(3), j\cdot}^m/\rho, \eta(1-\gamma) \sqrt{N}/\rho)
\end{align*}
with $S_G(x, \lambda) =  \max(1 - \lambda/||x||_2, 0)x$.


\pagebreak
\section{Starting values}
\label{sec:app_startingvalues}

As mentioned above the extreme solution for $\eta \rightarrow \infty$ depends on $\gamma$, $\bb{y}_k^{\infty}(\gamma)$. Here, we restate Proposition~\ref{proposition:extremes} and give the proof. 

\begin{proposition}
Using the notation of Equation~\eqref{eq:objective_matrix}, define $G_1(\bb{v}) = \sum_{j = 1}^p \sum_{i = 1}^N |v_{ji}| = \sum_{i = 1}^N ||\bb{v}_{\cdot i}||_1$ and $G_2(\bb{v}) = \sum_{j = 1}^p ||\bb{v}_{j \cdot}||_2$. 
\begin{itemize}
    \item[a.] For each source $i = 1, \ldots, N$ and given the normality constraint $||\bb{v}_{\cdot i}||_2 = 1$, the minimal value of $||\bb{v}_{\cdot i}||_1$ is attained, if there exists a variable $j'(i)$ such that $|v_{j'(i)i}| = 1$ and $v_{ji} = 0$ for all $j \neq j'(i)$. 
    \item[b.] Given the normality constraints $||\bb{v}_{\cdot i}||_2 = 1, i = 1, \ldots, N$, the minimal value of $G_2(\bb{v})$ is attained if there exists a variable~$j'$ for all sources $i = 1, \ldots, N$ such that $|v_{j'i}| = 1$ and $v_{ji} = 0$ for all~$j \neq j'$.
    \item[c.] The minimizers of $G_1(\bb{v})$ with the highest explained variance $\sum_{i = 1}^N \bb{v}_{\cdot i}' \bb{\hat{\Sigma}_i} \bb{v}_{\cdot i}$ have corresponding indices for non-zero entries $j'(i) \in \mathrm{arg} \max_{j = 1, \ldots, p}  \left(\bb{\hat{\Sigma}_i}\right)_{jj}$ for each source $i$. 
    The minimizers of $G_2(\bb{v})$ with highest explained variance have non-zero entries only for the variable indexed by~$j' \in \mathrm{arg} \max_{j = 1, \ldots, p} \sum_{i = 1}^N \left(\bb{\hat{\Sigma}_i}\right)_{jj}$.
\end{itemize}
\end{proposition}

\begin{proof}
\begin{itemize}
    \item [a.]  First, we know for any $\bb{v}$
        \begin{align*}
            1 = ||\bb{v}_{\cdot i}||_2^2 = \sum_{j = 1}^p v_{ji}^2 \le \sum_{j = 1}^p v_{ji}^2 + 2\sum_{j'<j} | v_{j'i}| | v_{ji}| =
            ||\bb{v}_{\cdot i}||_1^2
        \end{align*}
        and that the proposed minimizer of $||\bb{v}_{\cdot i}||_1$ has the minimal objective function value of 1. Moreover, all other minimizer have to fulfill that $| v_{j'i}| | v_{ji}| = 0$ for all~$j'<j$ to reach the minimal objective function value of 1. Thus, all minimizers have exactly one entry unequal to zero per source.
    \item [b.] Define $x_j = \sqrt{\sum_{i = 1}^N v_{ji}^2} = ||\bb{v}_{j \cdot}||_2$. Then, based on the inequality of part a, it holds that
    \begin{align*}
        N = \sum_{i = 1}^N ||\bb{v}_{\cdot i}||_2^2 = \sum_{i = 1}^N \sum_{j = 1}^p v_{ji}^2 = \sum_{j = 1}^p ||\bb{v}_{j \cdot}||_2^2 = ||x||_2^2 \leq ||x||_1^2 = \left(\sum_{j = 1}^p ||\bb{v}_{j \cdot}||_2 \right)^2=  G_2(\bb{v})^2.
        \end{align*}
        The extreme solution proposed above has an objective function value of $N$, which is thus minimal. Again, the same argument applies as before implies, that all mixed terms  $| x_{j'}| | x_{j}|$ must be equal to zero for all $j' \neq j$ to reach equality of the two norms.
    \item [c.] Trivial.
\end{itemize}
\end{proof}

For the special but important case of correlation matrices, we need an adaptation since the variable chosen by the explained variance is not unique. In order to ensure consistent behavior, we calculate the k-th eigenvectors of each correlation matrix, scale it with the root of the respective eigenvalue, and take the mean. The variable with the absolute highest value will be taken as the groupwise solution for $\eta \rightarrow \infty$. Although all variables are valid solutions for  $\eta \rightarrow \infty$, choosing an extreme solution close to $ \bb{y}_k^0$, we get more consistency over varying $\lambda$ and thus better convergence.

\subsection{Simulation results}
\label{subsec:app_startingvalues_simul}

In Figure~\ref{fig:starts_cov} the performance of the proposed starting value for simulation scenario 1 and the first four principal components is illustrated. We apply $\rho = p$, $\epsilon_{root} = 10^{-1}, \epsilon_{ADMM} = 10^{-4}, \epsilon_{thr} = 0.005$. For PCs 2 to 4 we iteratively use the best solution of all starting values for the orthogonality constraints. We simulate 100 random starting values, where each entry of the starting values is drawn from a standard normal distribution, and the vector is then projected onto the feasible space using $P_{\bar{V}}^1$. The values of the objective function of the random starting values are shown as boxplots for varying $\eta \in [0, 2]$ and $\gamma = 0, 0.5, 1$. The crosses indicate the objective function value for the proposed starting value. It is clearly visible, that the proposed starting value reliably produces optimal solutions and is thus a valid alternative to using many random starting values.

In Figure~\ref{fig:starts_cor} the results for correlation matrices with $\gamma = 0, 0.5, 1$ are shown. Regarding correlation matrices, there are multiple optimal extreme solutions, since all variables have the same amount of variance univariately. If there are multiple optimal solutions obtained in the simulations the one with the highest absolute scalar product with the extreme solution $ \bb{y}_k^0$ is taken for the orthogonality constraints for the simulation for the next PC. This accounts for the fact that the extreme solutions $\bb{y}_k^{\infty}(\gamma)$ (as well as $ \bb{y}_k^0$ for higher components) are based on the prior extreme solution components.

\begin{figure}
    \centering
    \includegraphics[width = \textwidth]{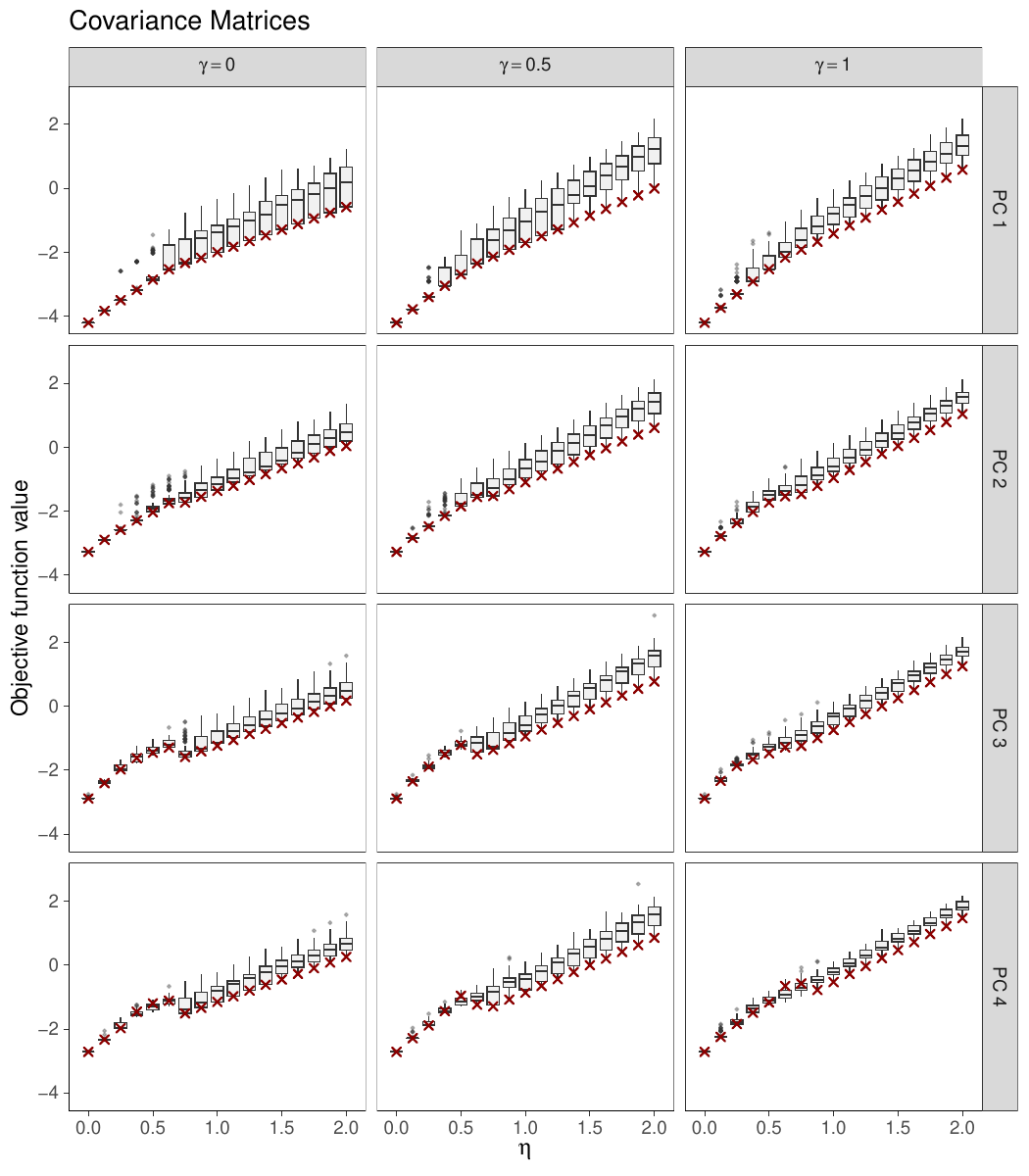}
    \caption{Objective function values for the proposed starting value (cross) and for random starting values (boxplot) for covariance matrices and $\gamma = 0.5$.}
    \label{fig:starts_cov}
\end{figure}

\begin{figure}
    \centering
    \includegraphics[width = 0.9\textwidth]{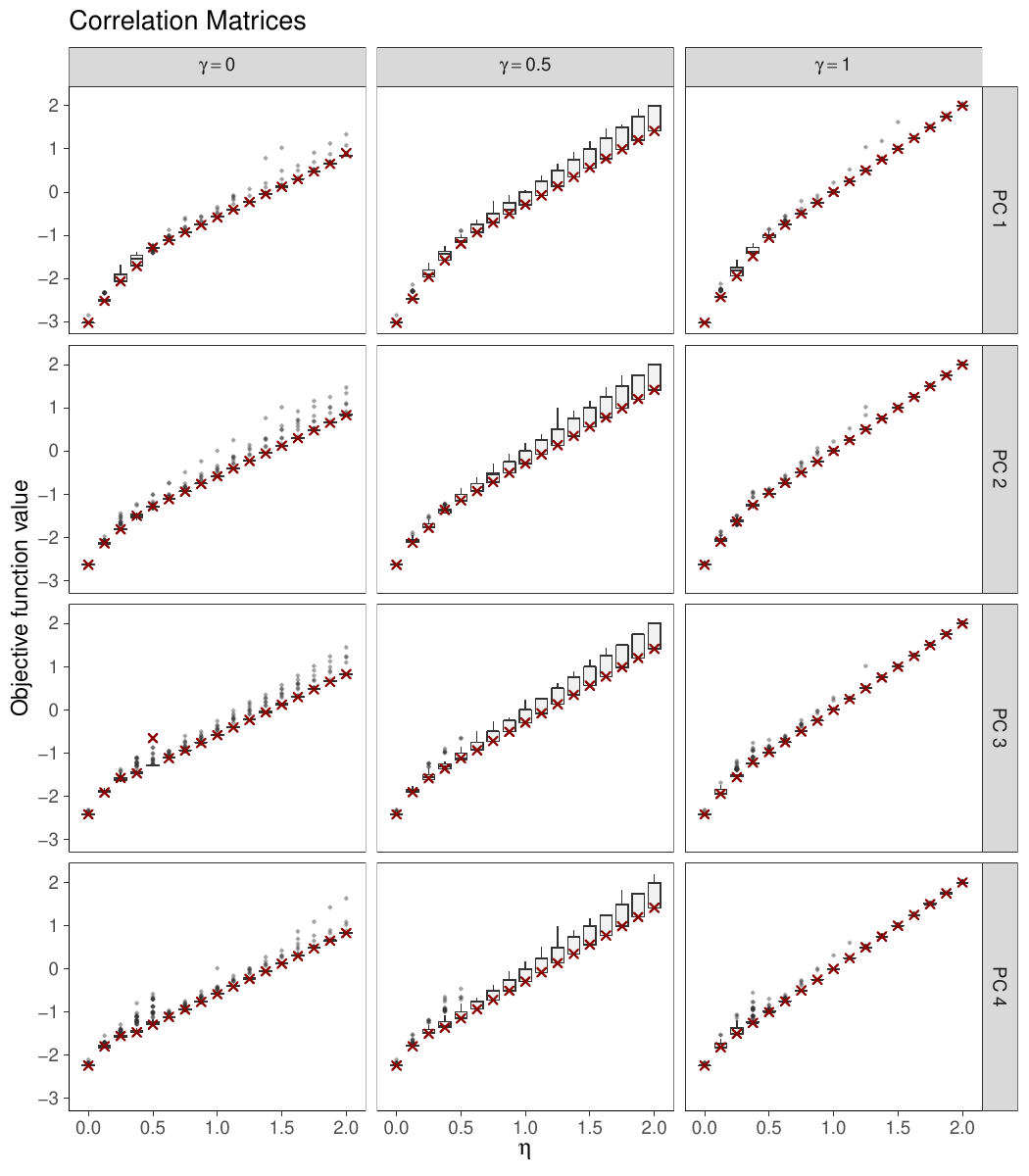}
    \caption{Objective function values for the proposed starting value (cross) and for random starting values (boxplot) for correlation matrices and $\gamma = 0.5$.}
    \label{fig:starts_cor}
\end{figure}

\pagebreak
\section{Weather Analysis at Hohe Warte}
\label{sec:app_weathervars}

The variables used in the real data example collected from the weather station Hohe Warte are listed and described in Table~\ref{tab:weather_vars}.

\begin{table}[H]
    \centering
    \begin{tabular}{|c|c|c|c|}
        \hline
        \textbf{Name} & \textbf{Description} & \textbf{} & \textbf{Unit} \\
        \hline
        cl & Cloud Coverage & Daily Mean & $1, \ldots, 100$ \\
        \hline
        rad & Global Radiation & Daily Sum & J/cm² \\
        \hline
        vp & Vapour Pressure & Daily Mean & hPa \\
        \hline
        wmax & Maximal Wind Speed & Daily Maximum & m/s \\
        \hline
        ap & Air Pressure & Daily Mean & hPa \\
        \hline
        hum & Relative Air Humidity & Daily Mean & \% \\
        \hline
        prec & Precipitation 24h & Daily Sum & mm \\
        \hline
        sight & Sight Distance & Sight at 1pm & m \\
        \hline
        sun & Sunshine Duration & Daily Sum & h \\
        \hline
        tmax & Maximal Air Temperature at 2m & Daily Maximum & °C \\
        \hline
        tmin & Minimal Air Temperature at 2m & Daily Minimum & °C \\
        \hline
        t & Air Temperature at 2m & Daily Mean & °C \\
        \hline
        w & Wind Speed & Daily Mean & m/s \\
        \hline
    \end{tabular}
    \caption{Variables of weather data example.}
    \label{tab:weather_vars}
\end{table}

\section{Geochemical Plant Analysis}
\label{sec:app_plants}

\begin{figure}[H]
    \centering
    \includegraphics[width=\textwidth]{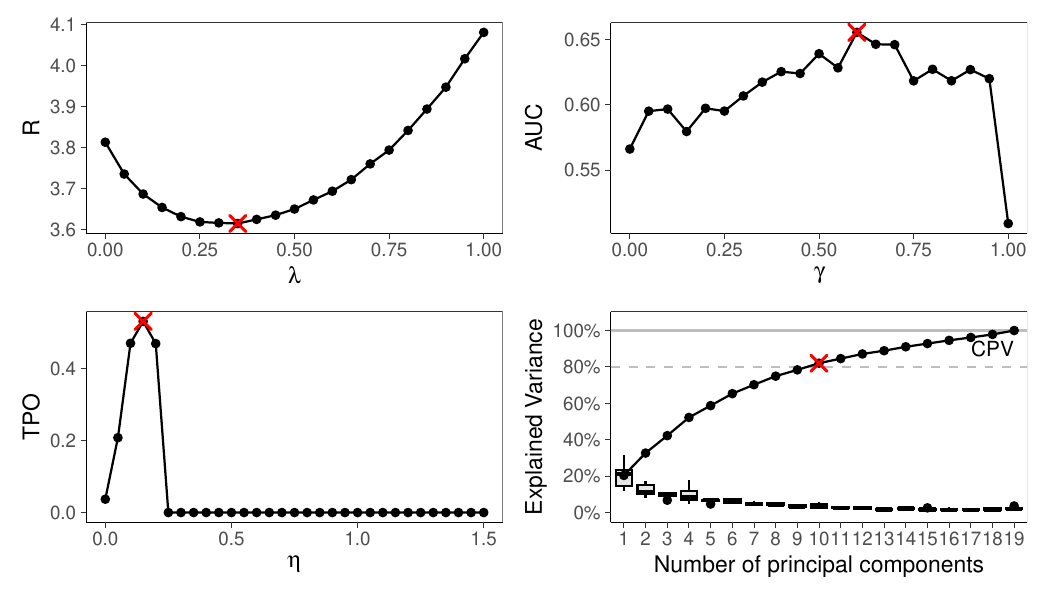}
    \caption{Plant geochemistry: Optimal smoothing and sparsity parameters.}
    \label{fig:plants_optpars}
\end{figure}

\end{document}